\documentclass{article}
\usepackage{color}
\usepackage{float}
\usepackage[left=1.25in, bottom=1.5in, right=1.5in, top=1in]{geometry}
\usepackage{authblk}

\usepackage{caption}
\usepackage{subcaption}
\usepackage[intlimits]{amsmath} \usepackage{amsxtra}
\usepackage{amsthm}
\usepackage{amssymb, eufrak, euscript} %
\usepackage{amsfonts}
\usepackage{graphicx}
\usepackage{rotating}
\usepackage{color}
\usepackage{epsfig}
\usepackage{verbatim}
\usepackage{algorithmic}
\usepackage{algorithm}
\usepackage{enumerate}
\usepackage{multirow}
\usepackage{array}
\usepackage{epstopdf}
\usepackage{pdfpages}
\usepackage{ulem}
\usepackage{soul}
\usepackage{cite}
\usepackage{ifpdf}
\usepackage{amsmath}
\usepackage{url}

\theoremstyle{plain}
\newtheorem{theorem}{Theorem}
\newtheorem{prop}[theorem]{Proposition}

\newtheorem{thm}{Theorem}[section]

\newtheorem{lem}[thm]{Lemma}

\theoremstyle{definition}

\theoremstyle{remark}

\newcommand{\argmin}{\operatornamewithlimits{argmin}}

\newcommand{\Var}{\operatornamewithlimits{Var}}
\newcommand{\Cov}{\operatornamewithlimits{Cov}}

\title{Low Rank  and Structured Modeling of High-dimensional Vector Autoregressions}

\author[1]{Sumanta Basu \footnote{These authors contributed equally} \thanks{Email: sumbose@cornell.edu}}
\author[2]{Xianqi Li $^*$ \thanks{Email: xianqili@ufl.edu}}
\author[3]{George Michailidis \thanks{Email: gmichail@ufl.edu}}
\affil[1]{Department of Statistical Science, Cornell University}
\affil[2]{ Department of Mathematics and the Informatics Institute, University of Florida}
\affil[3]{Departments of Statistics and Computer Science and the Informatics Institute, University of Florida}

\date{ }

\begin{document}
\maketitle

\begin{abstract}
\noindent
Network modeling of high-dimensional time series data is a key learning task due to its widespread use in a number of application areas, including macroeconomics, finance and neuroscience. While the problem of \textit{ sparse} modeling based on vector autoregressive models (VAR) has been investigated in depth
in the literature, more complex network structures that involve low rank and group sparse components  have received considerably less attention, despite their presence in data. Failure to account for low-rank structures results in spurious connectivity among the observed time series, which may lead practitioners to draw incorrect conclusions about pertinent scientific or policy questions. In order to accurately estimate a network of Granger causal interactions after accounting for latent effects, we introduce a novel approach for estimating \textit{ low-rank and structured sparse} high-dimensional VAR models. We introduce a regularized framework involving a combination of nuclear norm and lasso (or group lasso) penalty. Further, and subsequently establish non-asymptotic upper bounds on the estimation error rates of the low-rank and the structured sparse components. We also introduce a fast estimation algorithm and finally demonstrate  the performance of the proposed modeling framework over standard sparse VAR estimates through numerical experiments on synthetic and real data.
\end{abstract}


\section{Introduction}\label{sec:intro-lowrank}
The problem of learning the network structure among a large set of time series arises in many signal processing, economic, finance and biomedical applications. Examples include processing signals
obtained from radars \cite{Swindlehurst1998likelihood,Roman2000radar},
macroeconomic policy making and forecasting
\cite{lin2017regularized}, assessing connectivity among financial firms \cite{basu2017system},
reconstructing gene regulatory interactions from time-course genomic data \cite{michailidis2013autoregressive}
and understanding connectivity between brain regions from fMRI measurements \cite{van2010exploring}. Vector
Autoregressive (VAR) models provide a principled framework for these tasks.

Formally, a VAR model for $p$-dimensional time series $X_t$ is defined in its simplest possible form involving a single time-lag as
\begin{equation}\label{eqn:simple-VAR}
X^t=B'X^{t-1}+\epsilon^t,  \ \  \ t=1,\cdots, T,
\end{equation}
where $B$ is a $p\times p$ transition matrix specifying the lead-lag cross dependencies among the $p$ time series and $\{\epsilon^t\}$ is a zero mean error process. VAR models for small number of time series (low-dimensional) have been thoroughly studied in the literature \cite{lutkepohl2005new}. However, the above mentioned applications, where dozens to hundreds of time series are involved, created the need for the study of VAR models under high dimensional scaling and the assumption that their interactions are \textit{ sparse} to compensate for the possible lack of adequate number of time points (samples;
see \cite{basu2015regularized} and references therein).
There has been a growing body of literature on sparse estimation of large scale VAR models \cite{basu2014modeling}, including alternative penalties beyond the popular $\ell_1$ penalty (lasso), such as the
Berhu regularization introduced in \cite{he2013network},
group lasso type penalties employed in \cite{basu2015network,melnyk2016estimating}, as well as non-convex penalties akin to a square-lasso one investigated in \cite{jiang2017sparse}. Further, \cite{she2015network}
examine estimation of the transition matrix and the inverse covariance matrix of the error process through a joint sparse penalty. Note that the problem of sparse estimation of these two model parameters separately from a least squares and maximum likelihood viewpoints is addressed in \cite{basu2015regularized,lin2017regularized}, respectively, where in addition probabilistic finite sample error bounds for the obtained estimates are obtained.

Nevertheless, there are occasions where the sparsity assumption
may not be sufficient. For example, during financial crisis periods, returns on assets move together in a more concerted manner
\cite{basu2017system,billio2012econometric}, while transcription factors regulate a large number of genes that may lead to hub-node network structures \cite{tan2014learning}. Similarly, in brain connectivity networks, particular tasks activate a number of regions that cross talk in a collaborative manner \cite{sharaev2016effective}. Hence, it is of interest to study VAR models under high dimensional scaling where the transition matrix governing the temporal dynamics exhibits a more complex structure; e.g. it is \textit{ low rank} and/or \textit{ (group) sparse}.

In a low-dimensional regime, where the number of time
points scales to infinity, but the number of time series under study remains fixed, \cite{velu1986reduced} examined asymptotic properties of VAR models, where the parameters exhibit reduced rank structure and also discussed connections with canonical correlation analysis of such models presented in \cite{box1977canonical}. Specifically, the transition matrix $B$ in \eqref{eqn:model-sparse-lowrank} can be written as the product of two rank-$k$ matrices $\Phi, \Psi$, i.e. $B=\Phi \Psi'$, so that in the resulting model specification of the original $p$ time series is expressed as linear combinations $Z^t=\Psi X^t$ of the original ones, and $\Phi$ specifies the dependence between $X^t$ and $Z^t$; namely $X^t=\Phi' Z^{t-1}+\epsilon^t$. Hence, to obtain $\Phi$ and $\Psi$ \cite{velu1986reduced} suggest to estimate the parameters of the original model in \eqref{eqn:simple-VAR} under the constraint that $B=\Phi\Psi$ and that rank$(B)=k$. Other works include low rank approximations of Hankel matrices that represent the input-output structure of a linear time invariant systems and were studied in \cite{fazel2003log,chandrasekaran2011rank}. Finally, a brief mention to the possibility that the VAR transition matrix may exhibit such a structure appeared as a motivating example in \cite{agarwal2012}.

On the other hand, there is a mature literature on imposing low rank plus sparse, or pure group sparse structure for many
learning tasks for independent and identically distributed (i.i.d.) data. Examples include group sparsity in regression
and graphical modeling \cite{yang2017sparse+}, low rank and sparse matrix approximations for dimension reduction \cite{chandrasekaran2011rank}, etc. However, as shown in \cite{basu2015regularized}, the presence of temporal dependence across observations induces intricate dependencies between both rows and columns of the design matrix of the corresponding least squares estimation problem, as well as between the design matrix and the error term, that require careful handling to establish consistency properties for the model parameters under sparsity and high dimensional scaling. These issues are further compounded when more complex regularizing norms are involved, as discussed in \cite{melnyk2016estimating}.  In this paper, the authors model grouping structures within each column of $B$, but do not consider a low-rank component. In contrast, we focus on groups potentially spanning across different columns and allow a low-rank component in $B$.

{ More recently, \cite{zorzi2016ar} and \cite{zorzi2017sparse} extended the framework of  \cite{chandrasekaran2011rank} beyond decomposition of an observable matrix to Gaussian process identification by assuming a low-rank plus sparse structure on the inverse spectral density and the transfer function of a general VAR($d$) system, respectively. Our work is complementary to this recent developments. We directly model the transition matrix of a VAR(1) process that enables us to identify \textit{a directed network of} (group) sparse Granger causal relationships that are of interest in a number of applications; e.g. in  financial economics where firms with higher out-degree are of particular interest in measuring systemic risk \cite{billio2012econometric, basu2017system}. Further, our approach explicitly address identifiability issues for extracting the respective low-rank and sparse components, which in turn are leveraged to obtain probabilistic error bounds that characterize the quality of their estimates. The latter provide insights to the practitioner on sample size requirements and tuning parameter selection for real data applications. Finally, note that our approach to the issue of identifiability builds on \cite{agarwal2012}, wherein we characterize the degree of unidentifiability which guides in an explicit manner the selection of the tuning parameters used in the proposed optimization algorithm.}



Further, to estimate the posited model in \eqref{eqn:simple-VAR} with $B$ being low-rank and structured sparse (henceforth indicating that it could be either \textit{ pure sparse or group sparse or both}), we also introduce a fast accelerated proximal gradient algorithm, inspired by \cite{tseng2008accelerated,chen2017accelerated}, for the corresponding optimization problems. The key idea is that instead of searching for the local Lipschitz constant of the gradient of the smooth component of the objective function, the proposed algorithm utilizes a safeguarded Barzilai-Borwein (BB) initial stepsize \cite{barzilai1988two} and employs relaxed line search conditions to achieve better performance in practice. The latter enables the selection of  more ``aggressive" stepsizes,  while preserving the accelerated convergence rate of $\mathcal{O}(\frac{1}{k^2})$, where $k$ denotes the number of iterations required until convergence. Finally, the performance of the model parameters under different structures together with the associated estimation procedure based on the accelerated proximal gradient algorithm are calibrated on synthetic data, and illustrated on three data sets examining realized volatilities of  stock prices of $75$ large financial firms before, during and after the $2007-09$ US financial crisis.


\textit{\textbf{Notation:}} Throughout the paper, we employ the following notation: $\|.\|$, $\|.\|_2$ and $\|.\|_F$ denote the $\ell_2$-norm of a vector, the spectral norm and the Frobenius norm of a matrix, respectively. For a $p \times p$ matrix $B$,   the symbol $\|B\|_*$ is used to denote the nuclear norm, i.e. $\sum_{j=1}^{p}\sigma_{j}(B)$, the sum of the singular values of a matrix, while $B^\dagger$ denotes the conjugate transpose of a matrix $B$. For any matrix $B$, we use $\|B\|_0$ to denote $card(vec(B))$, $\|B\|_1$ for $\|vec(B)\|_1$ and $\|B\|_{\max}$ to denote $\|vec(B)\|_{\infty}$. Further, if $\{G_1, G_2, \ldots, G_K \}$ denote a partition of $\{1, 2, \cdots, p^2\}$ into $K$ non-overlapping groups, then we use $\|B\|_{2,1}$ to denote $\sum_{k=1}^{K}\|(B)_{G_k}\|_F$,  $\|B\|_{2,\max}$ for $\max_{k=1,2,...K}\|(B)_{G_k}\|_F$, while $\|B\|_{2,0}$ denotes the number of nonzero groups in $B$. Here, with a little abuse of notation, we use $B_{G_k}$ to denote $vec(B)_{G_k}$. In addition, $\Lambda_{\max}(.)$, $\Lambda_{\min}(.)$ denote the maximum and minimum eigenvalues of a symmetric or Hermitian matrix. For any integer $p \ge 1$, we use $\mathbb{S}^{p-1}$ to denote the unit ball $\{v \in \mathbb{R}^p : \|v\|=1\}$. We also use $\{e_1, e_2, \ldots \}$ generically to denote unit vectors in $\mathbb{R}^p$, when $p$ is clear from the context.  Finally, for positive real numbers $A, B$, we write $B \succsim A$ if there exists an absolute positive constant $c$, independent of the model parameters, such that $B \ge cA$.


\section{Model Formulation and Estimation Procedure}\label{sec:model-est}
Consider a VAR(1) model where the transition matrix $B$ is low-rank plus structured sparse given by
\begin{eqnarray}\label{eqn:model-sparse-lowrank}
&~& X^t = B^{'} X^{t-1} + \epsilon^t, ~~~~~~\epsilon^t \stackrel{i.i.d.}{\sim} N(0, \Sigma_\epsilon),\\
&~&~B = L^{*} +R^{*},~rank(L^{*})=r,
\end{eqnarray}
where $L^*$ corresponds to the low rank component and $R^*$ represents either a sparse $S^*$, or group-sparse component
$G^*$. It is further assumed that the number of non-zero elements in the sparse case is $\|S^{*}\|_0=s$, while in the group sparse case the number of non-zero groups is $\|G^{*}\|_{2,0}=g$, with $r \ll p, s \ll p^2$ and $g \ll p^2$.  The matrix $L^{*}$ captures persistence structure across \textit{ all} $p$ time series and enables the model to be applicable in settings where there are strong cross-autocorrelations, a feature that the standard sparse VAR model is not equipped to handle. The sparse or group sparse component captures additional cross-sectional autocorrelation structure among the time series. Finally, it is assumed that the error terms are serially uncorrelated. Our objective is to estimate $L^{*}$ and $R^{*}$ accurately based on a relatively small number of samples $N \ll p^2$.

\textit{Stability.} In order to ensure consistent estimation, we assume that the posited VAR model in \eqref{eqn:model-sparse-lowrank} is stable; i.e. its characteristic polynomial  $\mathcal{B}(z):= I_p - B^{'} z$ satisfies $\det(\mathcal{B}(z)) \neq 0$ on the unit circle of the complex plane $\{z \in \mathcal{C}: |z| = 1 \}$. This is a common assumption in the literature of multivariate time series \cite{lutkepohl2005new}, required for consistency and asymptotic normality of low-dimensional VAR models.  This assumption also ensures that the spectral density of the VAR model
\begin{equation}\label{eqn:spectral-density}
f_X(\theta) = \frac{1}{2\pi} \left(\mathcal{B}^{-1}\left(e^{i\theta}\right) \right) \Sigma_{\epsilon} \left(\mathcal{B}^{-1}\left(e^{i\theta}\right) \right)^\dagger, ~~~ \theta \in [-\pi, \pi],
\end{equation}
is bounded above in spectral norm.

It was shown in \cite{basu2015regularized} that this condition is sufficient to establish consistency of some regularized VAR estimates of a sparse transition matrix. Further, the following quantities play a central role in the error bounds of the regularized estimates:
\begin{eqnarray}\label{eqn:measure-stability}
\mathcal{M}(f_X) = \sup_{\theta \in [-\pi, \pi]} \Lambda_{\max}(f_X(\theta)), \nonumber \\
\EuFrak{m}(f_X) = \sup_{\theta \in [-\pi, \pi]} \Lambda_{\min} (f_X(\theta)), \nonumber \\
\mu_{\max}(\mathcal{B}) = \max_{|z| = 1} \Lambda_{\max} \left(\mathcal{B}^\dagger(z) \mathcal{B}(z) \right), \nonumber \\
\mu_{\min}(\mathcal{B}) = \min_{|z| = 1} \Lambda_{\min} \left(\mathcal{B}^\dagger(z) \mathcal{B}(z) \right).
\end{eqnarray}

As shown in \cite{basu2015regularized}, $\mathcal{M}(f_X)$ and $\EuFrak{m}(f_X)$ together capture the narrowness of the spectral density of a time series. Processes with stronger temporal and cross-sectional dependence have narrower spectra that in turn lead to slower convergence rates for the regularized estimates. For VAR models, $\mathcal{M}(f_X)$ and $\EuFrak{m}(f_X)$ are related to $\mu_{\max}(\mathcal{B})$ and $\mu_{\min}(\mathcal{B})$ as follows:
\begin{equation}\label{eqn:bound-stability-measures}
\EuFrak{m}(f_X) \ge \frac{1}{2\pi} \frac{\Lambda_{\min}(\Sigma_{\epsilon})}{\mu_{\max}(\mathcal{B})}, ~~~
\mathcal{M}(f_X) \le \frac{1}{2\pi} \frac{\Lambda_{\max}(\Sigma_{\epsilon})}{\mu_{\min}(\mathcal{B})}.
\end{equation}

Proposition 2.2 in \cite{basu2015regularized} provides a lower bound on $\mu_{\min}(\mathcal{B})$. For the special structure of the models considered here, we can get an improved upper bound on $\mu_{\max}(\mathcal{B})$, as shown in the following lemma:
\begin{lem}
For a stable VAR(1) model of the class \eqref{eqn:model-sparse-lowrank}, we have
\begin{equation}
\mu_{\max}(\mathcal{B}) \le \left[1 + l + (v_{in}+v_{out})/2 \right]^2
\end{equation}
 where $l$ is the largest singular value of $L^{*}$, $v_{in} = \max_{1 \le j \le p} \sum_{i=1}^p |R^{*}_{ij}|$ and $v_{out} = \max_{1 \le i \le p} \sum_{j=1}^p |R^{*}_{ij}|$.
\end{lem}
\begin{proof}
$\|\mathcal{B}(z)\| = \|I - (L^{*}+R^{*})z \| \le \|I\| + \|L^{*}\| + \|R^{*}\|$ for any $z \in \mathbb{C}$ with $|z|=1$. The result follows from the fact that $\mu_{\max}(\mathcal{B})  = \max_{|z|=1} \|\mathcal{B}(z)\|^2$.
\end{proof}

\subsection{Estimation Procedure}
The estimation of VAR model parameters is based on the following regression formulation (see \cite{lutkepohl2005new}). Given $T+1$ consecutive observations $\{X^0, \cdots, X^{T}\}$ from the VAR model, we work with the autoregressive design as follows:
\begin{eqnarray}\label{eqn:data-sparse-lowrank}
\underbrace{\left[\begin{array}{c} (X^T)' \\ \vdots \\ (X^{1})' \end{array} \right]}_{\mathcal{Y}}
& = & \underbrace{\left[ \begin{array}{c}(X^{T-1})'\\
						    \vdots \\
						 (X^{0})\end{array} \right]}_{\mathcal{X}}
	B
	+ \underbrace{\left[ \begin{array}{c} (\epsilon^T)' \\ \vdots \\ (\epsilon^1)' \end{array}\right]}_{E}.
\end{eqnarray}
This is a standard regression problem with $N\equiv T$ samples and $q = p^2$ variables. Our goal is to estimate $L^{*}$, $R^{*}$ with high accuracy when $N \ll p^2$.

There is an \textit{inherent identifiability} issue in the estimation of the components $L^*$ and $R^*$. Suppose the low-rank component $L^{*}$ itself is $s$-sparse or $g$-group sparse and the sparse or group-sparse component $R^{*}$ is of rank $r$. In that scenario, we can not hope for any method to estimate $L^{*}$ and $R^{*}$ separately  without imposing any further constraints. So, a minimal condition for low-rank and sparse or group-sparse recovery is that the low rank part should not be too sparse and the sparse or group-sparse part should not be low-rank.

This issue has been rigorously addressed in the literature
(e.g. \cite{chandrasekaran2011rank})
for independent and identically distributed data and resolved by imposing an \textit{ incoherence} condition. Such a condition is sufficient for \textit{ exact} recovery of the low rank and the sparse or group-sparse component by solving a convex program. In a recent paper, \cite{agarwal2012} showed that in a noisy setting where exact recovery of the two components is impossible, it is still possible to achieve good estimation error under a comparatively mild assumption. In particular, they formulated a general measure for the \textit{radius of nonidentifiability} of the problem under consideration and established a non-asymptotic upper bound on the estimation error $\|\hat{L} - L^{*} \|^2_F + \| \hat{R} - R^{*} \|^2_F$, which depend on this radius. The key idea is to allow for simultaneously sparse (or group-sparse) and low-rank matrices in the model, and control for the error introduced. We refer the readers to the above paper for a more detailed discussion on this notion of nonidentifiability. In this work, the low-rank plus sparse or group-sparse decomposition problem under restrictions on the radius of nonidentifiability takes the form
\begin{equation*}
(\hat{L}, \hat{R}) = \argmin_{\substack{
   L,R\in \mathbb{R}^{p \times p}  \\
   L \in \Omega
  }} l(L, R),
\end{equation*}
\begin{equation}\label{eqn:opt-sparse-lowrank}
l(L, R) := \frac{1}{2} \left\| \mathcal{Y} - \mathcal{X} (L+R) \right\|^2_F + \lambda_N \|L\|_* + \mu_N \|R\|_{\diamond}.
\end{equation}
Here $\Omega = \left\{L\in \mathbb{R}^{p \times p}: \|L\|_{\max}\leq \frac{\alpha}{p}\right\}$ (for sparse) or $\left\{L\in \mathbb{R}^{p \times p}: \|L\|_{2,\max}\leq \frac{\beta}{\sqrt{K}}\right\}$ (for group sparse), $\|\cdot\|_{\diamond}$ represents $\|\cdot\|_1$ or $\|\cdot\|_{2,1}$ depending on sparsity or group sparsity of $R$, and $\lambda_N$ and $\mu_N$ are non-negative tuning parameters controlling the regularizations of low-rank and sparse/group-sparse parts. { The parameters $\alpha$ and $\beta$ control for the degree of nonidentifiability of the matrices allowed in the model class.  For instance, larger values of $\alpha$ provide sparser estimates of $S$ and allow simultaneously sparse and low-rank components to be absorbed in $\hat{L}$. A smaller value of $\alpha$, on the other hand, tends to produce a matrix $L$ with smaller and pushes the simultaneously low-rank and sparse components to be absorbed in $\hat{S}$.
In practice, we recommend choosing $\alpha$ and $\beta$ in the range $[1,p]$ and $[1,K]$, respectively. The issue of selecting them robustly in practice is discussed in Section \ref{sec:num-exps}.}

\textit{Remark.}
On certain occasions, it may be useful to have both sparse \textit{m and} group-sparse structures in the model, in addition to the low rank structure. We then have
$R^* = S^* + G^*$ in \eqref{eqn:opt-sparse-lowrank} with $\|R\|_{\diamond} =  \|S\|_1 + \frac{\nu_N}{\mu_N}\|G\|_{2,1}$. However, to guarantee the \textit{ simultaneous identifiability}
of the sparse and group-sparse components from the low-rank component, stronger conditions need to be imposed
on $L$; namely, $\Omega = \left\{L\in \mathbb{R}^{p \times p}: \|L\|_{\max}\leq \frac{\alpha}{p} ~~\&~~ \|L\|_{2,\max}\leq \frac{\beta}{\sqrt{K}}\right\}$.

{ \textit{Remark.} Note that the estimated VAR model is not guaranteed to be stable, although the error bound analysis in section \ref{sec:theory} ensures its stability with high probability as long as the sample size is large enough and the true generative model is stable. For network reconstruction and visualization purposes, stability of the estimated VAR is not strictly required. However, enforcing stability is essential for forecasting purposes. When there is a small deviation of the estimated model from stability (e.g the spectral radius of the estimated $\hat{B}$ is a little over $1$), stability can be ensured through a post-processing step of shrinking the moduli of  eigenvalues of $\hat{B}$ below  $1$ while keeping its eigenvectors unchanged. This type of projection argument is common in covariance and correlation matrix estimation with missing data for ensuring positive definiteness of the estimates  \cite{higham2002computing}. However, in case of moderate to large deviation from stability, a closer look at the individual time series is recommended to re-assess the validity of the VAR formulation. For example, in macroeconomics, it is customary to use suitable transformations of the component time series to ensure that each of the individual time series and the resulting VAR model is stable, as opposed to modeling the individual and the joint time series (without transformation) as unit root and co-integrating processes.   For instance, see \cite{banburra2010large, giannone2015prior} for specific recommendations on useful transformations for macroeconomic time series.}

\section{Theoretical Properties}\label{sec:theory}
Next, we derive non-asymptotic upper bounds on the estimation errors of the low-rank plus structured sparse components of the transition matrix $B$. The main result shows that consistent estimation is possible with a sample size of the order $N \asymp p \, \mathcal{M}^2(f_X)/\EuFrak{m}^2(f_X)$, as long as the process $\{X^t\}$ is stable and the radius of nonidentifiability, as measured by $\|L^*\|_{\max}$ and/or $\|L^*\|_{2,\max}$, is small in an appropriate sense detailed next.

To establish the results, we first consider fixed realizations of $\mathcal{X}$ and $E$ and impose the
following assumptions: \\
1) \textit{Restricted Strong Convexity (RSC)}: There exist $\zeta > 0$ and $\tau_N > 0$ such that
$$
\frac{1}{2}\|\mathcal{X}\Delta\|_{\text{F}}^2 \geq \frac{\zeta}{2}\|\Delta\|_{\text{F}}^2-\tau_N\Phi^2(\Delta), \ \ \text{for all} \ \ \Delta \in \Re^{p\times p}
$$
where $\Phi(\Delta)=\underset{L+R=\Delta}{\text{inf}}\{\lambda_N \|L\|_* + \mu_N \|R\|_{\diamond}\}$,
and \\
2) \textit{Deviation Conditions}: There exist a deterministic function $\phi(B, \Sigma_\epsilon)$ of the model parameters $B$ and $\Sigma_\epsilon$ such that
$$
\|\mathcal{X}'E/N\|_{2} \leq \phi(B, \Sigma_{\epsilon}) \sqrt{p/N}, \ \ \text{and}
$$
$$
\|\mathcal{X}'E/N\|_{\max} \leq \phi(B,\Sigma_{\epsilon})\sqrt{\frac{2\log p}{N}}, \ \ \text{and}
$$
$$
\|\mathcal{X}'E/N\|_{2,\max} \leq \phi(B,\Sigma_{\epsilon})\frac{\sqrt{m \log K}}{\sqrt{N}},
$$
where $m$ is the size of the largest group  $\max_{1 \le k \le K} card(G_k)$. 

Later on, we show that assumptions 1) and 2) are indeed satisfied with high probability when the data are generated from the model \eqref{eqn:model-sparse-lowrank}.

Next, we present the non-asymptotic upper bounds on the estimation errors of the low-rank plus structured sparse components, respectively.
\begin{prop}\label{CLplusS}
(a)
Suppose that the matrix $L^*$ has rank at most $r$, while the matrix $S^*$ has at most $s$ nonzero entries. Then, for any $\lambda_N \geq 4\|\mathcal{X}^{'}E\|_{2}$ and $\mu_N \geq 4\|\mathcal{X}^{'}E\|_{\max}+\frac{4\zeta\alpha }{p}$, any solution $(\hat{L}, \hat{S})$ of \eqref{eqn:opt-sparse-lowrank} satisfies
$$
\|\hat{L}-L^{*}\|_F^2+\|\hat{S}-S^{*}\|_F^2  \leq
\frac{4}{\zeta^2}(\frac{9}{2}\lambda_N^2r+4\mu_N^2s).
$$
(b)
Suppose that the matrix $L^*$ has rank at most $r$, while the matrix $G^*$ has at most $g$ non-zero groups. Then, for any $\lambda_N \geq 4\|\mathcal{X}^{'}E\|_2$ and $\mu_N \geq 4\|\mathcal{X}^{'}E\|_{2,\max}+\frac{4\zeta\beta }{\sqrt{K}}$, any solution $(\hat{L}, \hat{G})$ of \eqref{eqn:opt-sparse-lowrank} satisfies
$$
\|\hat{L}-L^{*}\|_F^2+\|\hat{G}-G^{*}\|_F^2  \leq
\frac{4}{\zeta^2}(\frac{9}{2}\lambda_N^2r+4\mu_N^2g)
$$
\end{prop}
\textit{Remark.} It should be noted that if each group in $G^*$ has only one element, then we have $K=p^2$ and $g$ non-zero entries. For such cases, part (b) of Proposition \ref{CLplusS} becomes identical to part (a).


As a byproduct,  we also give the estimation error bound of the transition matrix which can be characterized by  the sparse plus group-sparse and the low-rank plus sparse and group-sparse components, respectively, under the assumption that the strength of the connections in the group-sparse component $G$ is weak; i.e. $G \in \Psi$ with $\Psi = \left\{G \in \mathbb{R}^{p \times p}: \|G\|_{\max}\leq \frac{\gamma}{p}\right\}$, where $\gamma \in [1, p]$.
\begin{prop}\label{CSplusGS}
(a) Suppose that the matrix $S^*$ has at most $s$ nonzero entries, while the matrix $G^*$ has at most $g$ non-zero groups. Then, for any $\mu_N \geq 4\|\mathcal{X}^{'}E\|_{\max}+\frac{4\zeta\gamma }{p}$ and $\nu_N \geq 4\|\mathcal{X}^{'}E\|_{2,\max}$, any solution $(\hat{S}, \hat{G})$ of \eqref{eqn:opt-sparse-lowrank} satisfies
$$
\|\hat{S}+\hat{G}-S^{*}-G^{*}\|_F^2  \leq
\frac{4}{\zeta^2}(8\mu_N^2s+9\nu_N^2g).
$$
(b) Suppose that the matrix $L$ has rank at most $r$, while the matrix $S$ has at most $s$ nonzero entries and the matrix $G$ has at most $g$ non-zero groups. Then, for any $\lambda_N \geq 4\|\mathcal{X}^{'}E\|_2$, $\mu_N \geq 4\|\mathcal{X}^{'}E\|_{\max}+\frac{4\zeta\alpha }{p}+\frac{4\zeta\gamma}{p}$, and $\nu_N\geq 4\|\mathcal{X}^{'}E\|_{2,\max}+\frac{4\zeta\beta}{\sqrt{K}}$,  any solution $(\hat{L}, \hat{S}, \hat{G})$ of \eqref{eqn:opt-sparse-lowrank} satisfies
$$
\|\hat{L}-L^*\|_F^2+\|\hat{S}+\hat{G}-S^{*}-G^{*}\|_F^2 \leq
\frac{4}{\zeta^2}(9\lambda_N^2r+\frac{25}{2}\mu_N^2s+8\nu_N^2g).
$$
\end{prop}
\noindent See Appendix A for the detailed proof of Propositions \ref{CLplusS} and \ref{CSplusGS}.

Note that the objective in Proposition \ref{CSplusGS}
is not the accurate recovery of the $S^*$ and $G^*$ components \textit{separately}. The latter can be in principle achieved, if one sets $\gamma$ to a very small value.

In order to obtain meaningful results in the context of our problem, we need  upper bounds on $\|\mathcal{X}'E\|_2$, $\|\mathcal{X}'E \|_{\max}$ and $\|\mathcal{X}'E \|_{2,\max}$ and a lower bound on $\Lambda_{\min}(\mathcal{X}'\mathcal{X})$ that hold with high probability. For the case of independent and identically distributed data, such high-probability deviation bounds are established in \cite{agarwal2012}. However, for time series data all entries of the  $\mathcal{X}$ matrix are dependent on each other, and hence it is a non-trivial technical task to establish such deviation bounds. A key technical contribution of this work is to derive these deviation bounds, which lead to meaningful analysis for VAR models. The results rely on the measure of stability defined in \eqref{eqn:measure-stability} and an analysis of the joint spectrum of $\{X^{t-1}\}$ and $\{\epsilon^t\}$ undertaken next.

\begin{prop}\label{prop:conc-sparse-lowrank}
Consider a random realization of $\{X^0, \ldots, X^T \}$ generated according to a stable VAR(1) process \eqref{eqn:model-sparse-lowrank} and form the autoregressive design \eqref{eqn:data-sparse-lowrank}. Define
\begin{eqnarray*}
\phi(B, \Sigma_{\epsilon}) = \Lambda_{\max}(\Sigma_{\epsilon}) \left[1+ \frac{1+\mu_{\max}(\mathcal{B})}{\mu_{\min}(\mathcal{B})} \right]
\end{eqnarray*}
Then, there exist universal positive constants $c_i > 0$ such that
\begin{enumerate}
\item for $N \succsim p$,
\begin{equation*}
\mathbb{P}\left[ \| \frac{\mathcal{X}'E}{N}\|_{2} > c_0 \phi(B, \Sigma_{\epsilon}) \sqrt{\frac{p}{N}} \right] \le c_1 \exp \left[-c_2 \log p \right]
\end{equation*}
and for any $N \succsim \log p$,
\begin{equation*}
\mathbb{P}\left[ \| \frac{\mathcal{X}'E}{N}\|_{\max} > c_0 \phi(B, \Sigma_{\epsilon}) \sqrt{\frac{\log p}{N}} \right]
\le c_1 \exp \left[-c_2 \log p \right]
\end{equation*}
and for $N \succsim m \log p$,
\begin{equation*}
\mathbb{P}\left[ \| \frac{\mathcal{X}'E}{N}\|_{2,\max} > c_0 \phi(B, \Sigma_{\epsilon}) \frac{\sqrt{m \log p}}{\sqrt{N}} \right]
\le c_1 \exp \left[-c_2\log p \right]
\end{equation*}
\item for $N \succsim p \mathcal{M}^2(f_X)/\EuFrak{m}^2(f_X)$,
\begin{equation*}
\mathbb{P}\left[ \Lambda_{\min}(\frac{\mathcal{X}'\mathcal{X}}{N}) > \frac{\Lambda_{\min}(\Sigma_{\epsilon})}{2 \mu_{\max}(\mathcal{B})} \right] \le c_1 \exp \left[-c_2 \log p \right]
\end{equation*}
\end{enumerate}
\end{prop}
\noindent See  Appendix B for the  detailed  proof  of  Proposition \ref{prop:conc-sparse-lowrank}.

Using the above deviation bounds in the non-asymptotic errors of Propositions \ref{CLplusS}, 
we obtain the final  result for approximate recovery of the low-rank and the structured sparse components using nuclear and $\ell_1/\ell_{2,1}$ norm relaxations, as we show next.
\begin{prop}\label{prop:main-result-low-rank}
Consider the setup of Proposition \ref{prop:conc-sparse-lowrank}. There exist universal positive constants $c_i>0$ such that for $N \succsim p \mathcal{M}^2(f_X)/\EuFrak{m}^2(f_X)$, and $\|L^{*}\|_{\max} \le \alpha/p$, any solution $(\hat{L}, \hat{S})$ of the program \eqref{eqn:opt-sparse-lowrank} satisfies, with probability at least $1-c_1 \exp[-c_2 \log p]$,
\begin{equation}
\|\hat{S}-S^{*}\|^2_F + \|\hat{L}-L^{*}\|^2_F \le \frac{c_0 \phi^2(B, \Sigma_{\epsilon}) \mu_{\max}^2(\mathcal{B})}{\Lambda_{\min}^2(\Sigma_\epsilon)}
 \frac{(rp + s \log p)}{N} + \frac{32 \Lambda^2_{\min}(\Sigma_{\epsilon})}{\mu^2_{\max}(\mathcal{B})}\frac{s \alpha^2}{p^2}.
\end{equation}
\end{prop}
\textit{Remark:} The error bound presented in the above proposition consists of two key terms. The first term is the error of estimation emanating from randomness in the data and limited sample capacity. For a given model, this error goes to zero as the sample size increases. The second term represents the error due to the unidentifiability of the problem. This is more fundamental to the structure of the true low-rank and structured sparse components, and depends only on the model parameters and does not vanish, even with infinite sample size.

Further, the estimation error is a product of two terms - the second term $(rp+s \log p)/N$ involves the dimensionality parameters and matches the parametric convergence rate for independent observations. The effect of dependence in the data is captured through the first part of the term: $ \frac{c_0 \phi^2(B, \Sigma_{\epsilon}) \mu_{\max}^2(\mathcal{B})}{\Lambda_{\min}^2(\Sigma_\epsilon)}$. As discussed in \cite{basu2015regularized}, this term is larger when the spectral density is more spiky, indicating a stronger temporal and cross-sectional dependence in the data.

\begin{prop}\label{prop:main-result-low-rank+G}
Consider the setup of Proposition \ref{prop:conc-sparse-lowrank}. There exist universal positive constants $c_i>0$ such that for $N \succsim p \mathcal{M}^2(f_X)/\EuFrak{m}^2(f_X)$, for any $G^0$ with $\|L^{*}\|_{2,\max} \le \beta/\sqrt{K}$, any solution $(\hat{L}, \hat{G})$ of the program \eqref{eqn:opt-sparse-lowrank} satisfies, with probability at least $1-c_1 \exp[-c_2 \log p]$,
\begin{equation}
\|\hat{G}-G^{*}\|^2_F + \|\hat{L}-L^{*}\|^2_F  \le \frac{c_0 \phi^2(B, \Sigma_{\epsilon}) \mu_{\max}^2(\mathcal{B})}{\Lambda_{\min}^2(\Sigma_\epsilon)}
\frac{(rp + g (m\log p))}{N} + \frac{32 \Lambda^2_{\min}(\Sigma_{\epsilon})}{\mu^2_{\max}(\mathcal{B})}\frac{g \beta^2}{K}.
\end{equation}
\end{prop}
\noindent See  Appendix B for  the  detailed  proof  of  Proposition \ref{prop:main-result-low-rank} and \ref{prop:main-result-low-rank+G}.

\textit{Remark.} Based on Proposition \ref{prop:main-result-low-rank+G}, similar conclusions can be obtained as that for the low rank plus sparse case.

\section{Computational Algorithm and its Convergence Properties}\label{sec:algorithm}
Next, we introduce a fast algorithm for estimating the transition matrix $B$ from data. For ease of presentation and to convey the key ideas clearly, we first present the algorithm for $B$ representing a single structure (e.g. only low rank, or only group sparse, or only sparse), and in addition establish its convergence properties.
Subsequently, we modify the algorithm to handle the composite structures considered in this paper and also establish its convergence.

The fast network structure learning (FNSL) Algorithm \ref{alg:fnsl} is described next.
A safeguarded BB initial value is selected, as the initial choice of the nominal step $\eta_i$, i.e.
\begin{equation}\label{eq30}
\eta_{0,i} = \max \left \{\eta_{\min}, \frac{\|\mathcal{X}(B_i-B_{i-1})\|^2_F}{\|B_i-B_{i-1}\|^2_F}\right\} \ \ \text{for} \ \ i> 1.
\end{equation}
For notational convenience, the penalty term for estimating the transition matrix $B$ is denoted by $P_B(B, \lambda)$, where $\lambda > 0$ represents the tuning parameter.
\begin{algorithm*}
	\caption{ Fast Network Structure Learning (FNSL) method}
    \label{alg:fnsl}
	\begin{algorithmic}
		\STATE Choose $C\geq 0, \sigma > 1, \eta_{0,1} \geq\eta_{\min}$.
               Set $\alpha_1=1, B_{1}^{ag}=B_1$, and $Q_1=0$.
		\STATE \textbf{For} $i = 1,2, \ldots, k$,
		\begin{enumerate}[\hspace{.2cm}1.\hspace{.5cm}]
			\item[] \slash\slash {\it { Backtracking} }
			\item Set $\eta_i=\alpha_i\eta_{0,i}$, where $\eta_{0,i}$ is from \eqref{eq30}. Solve $\alpha_{i}$ from
                 $\frac{1}{\alpha_{i-1}\eta_{i-1}}=\frac{1-\alpha_{i}}{\alpha_{i}\eta_{i}}$ for $i>1$. Compute
			\begin{align*}
			B_i^{md}     = & (1 - \alpha_i)B_i^{ag} + \alpha_i B_i,
			\\
			B_{i+1} = & \underset{B}{\arg\min}\left\{\langle \nabla l(B_i^{md}), B\rangle+\frac{\eta_i}{2}\|B-B_i\|^{2}_F
			+ P_B(B, \lambda) \right\},
			\\
			\Gamma_i=&\|B_{i+1}-B_{i}\|^2-\frac{\alpha_i}{\eta_i}\|\mathcal{X}(B_{i+1}-B_{i})\|^2_F,
			\\
			Q_{i+1}=& \beta_iQ_i+\Gamma_i,\ \text{ where }0\leq \beta_i \leq (1-\frac{1}{i})^2.
			\end{align*}
			\item
			If $Q_{i+1}< -{C}/{i^2}$, then replace $\eta_{0,i}$ by $\sigma\eta_{0,i}$ and return to step 1.
			
			\item[] \slash\slash {\it { Updating iterates} }
			\item Compute
            \begin{align*}
            B_{i+1}^{ag} = (1-\alpha_i)B_i^{ag}+\alpha_iB_{i+1}.
			\end{align*}
		\end{enumerate}
		\STATE \textbf{EndFor}
		\STATE \textbf{Output} $B_{k+1}^{ag}$.
	\end{algorithmic}
\end{algorithm*}
The specific $B_{i+1}$ update depends on the employed penalty term; for an $\ell_1$ penalty inducing sparsity, it corresponds to soft-thresholding
\cite{daubechies2004iterative}, for a group sparse penalty to group soft-thresholding \cite{yang2017sparse+}, while for a nuclear norm penalty to singular value thresholding \cite{cai1956singular}.

It can also been seen in  Algorithm \ref{alg:fnsl}, that for $\alpha_i \equiv 1$ for $\forall i\geq 1$, then $B_i^{md} = B_i$ and $B_{i+1}^{ag} = B_{i+1}$,  which leads to the traditional gradient descent algorithm. Indeed, Algorithm \ref{alg:fnsl} is obtained by incorporating an efficient backtracking strategy into the accelerated multi-step scheme by \cite{nesterov2004introductory,tseng2008accelerated}. It provides a different way to look for a larger stepsize by employing a relaxed line search condition, instead of searching for the gradient Lipschitz constant of the data fidelity term. { Steps 1 and 2 constitute the backtracking ones. Both $B_i^{md}$ and $B_i^{ag}$ are linear combinations of all past iterations of $B_i$, but based on different weights as can be seen from their updates, i.e.  $B_i^{ag}=(1-\alpha_{i-1})B_{i-1}^{ag}+\alpha_{i-1} B_i$ and $B_i^{md}=(1-\alpha_i)B_{i}^{ag}+\alpha_i B_i$. Here `ag' simply denotes `aggregate'.  The data fidelity term in the cost function is linearized at $B_i^{md}$. Since it is used after we have obtained $B_i^{ag}$ and before we obtain  $B_{i+1}$,  we use `md' to denote a `middle' update. Further, $\Gamma_i$ and $Q_i$ are the parameters most closely related to our line search conditions. Intuitively, if we set $\Gamma_i \geq 0$, we are  certainly able to guarantee the accelerated rate of convergence. However, this will render the stepsize smaller. Fortunately, our convergence analysis enables us to relax this condition by utilizing the summing up procedure, with $Q_i$ corresponding to the part of the sum of $\Gamma_i$. Thus, we can impose a relaxed termination condition on $Q_i$ (see step 2 of Algorithm \ref{alg:fnsl}) without impacting the rate of convergence while being able to obtain a more aggressive stepsize}. In fact, parameter $C$ in Step 2 plays an important role, i.e. the number of
the trial steps can be reduced significantly when a relatively larger $C$ is selected. However, the value of $C$ can not be too large either, since it might impair the convergence rate in terms of the objective function value.


The convergence rate of the proposed algorithm \ref{alg:fnsl} is established next.
\begin{prop}\label{ACRNL}
Let $\{B_{k+1}^{ag}\}$ be generated by Algorithm \ref{alg:fnsl}. Then, for any $k \geq 1$
\begin{equation}\label{TR1}
\begin{array}{ll}
l(B_{k+1}^{ag}) - l(\hat{B}) \leq \frac{2\sigma\|\mathcal{X}\|_2^2\|B_0 - \hat{B}\|^2_F+\tilde{C}}{(k+1)^2}
\end{array}
\end{equation}
where $\tilde{C}$ is a finite positive number independent of $k$.
\end{prop}
\noindent See Appendix C for the detailed proof. It should be noted that the convergence rate of $\{B_{k+1}\}$ is still an open problem, which needs to be addressed further in future, considering its good performance for some cases.

Next, we enhance the algorithm for solving \eqref{eqn:opt-sparse-lowrank} in the general case. The accelerated convergence rate can be obtained by following the proof for Proposition \ref{ACRNL}.
Similarly to the case in Algorithm \ref{alg:fnsl}, the initial trial step of $\eta_i$ is a safeguarded BB choice
\begin{equation}\label{eq31}
\eta_{0,i} = \max \left \{\eta_{\min}, \frac{\|\mathcal{X}(L_i+R_i-L_{i-1}-R_{i-1})\|^2_F}
{\|L_i+R_i-L_{i-1}-R_{i-1}\|^2_F}\right\}
\end{equation}

The update of the $L$ component is based on singular value thresholding, while that of the $R$ component on (group) soft-thresholding.
Note that the most expensive computational operation corresponds to the singular value decomposition (SVD) when updating
$L_{i+1}$. As mentioned earlier, the proposed algorithm is able to look for larger magnitude step sizes by conducting fewer number of line searches,
due to employing more relaxed line search conditions. Actually, this is an important improvement considering the computational
cost of SVD. Indeed, the efficiency of the proposed algorithm can be enhanced further if we employ
the truncated SVD \cite{lin2011some} instead of the full SVD.

The convergence rate of the proposed algorithm 2 (given in supplement due to limited space) is established next.
\begin{prop}\label{ACRNL2}
Let $(L_{k+1}^{ag}, R_{k+1}^{ag})$ be a sequence of updates generated by Algorithm 2. Then, for any $k \geq 1$
\begin{equation}
\begin{array}{ll}
l(L_{k+1}^{ag}, R_{k+1}^{ag}) - l(\hat{L}, \hat{R}) \leq
\frac{2\sigma\|\mathcal{X}\|_2^2\big(\|L_0 -\hat{L}\|^2_F+\|R_0 - \hat{R}\|^2_F
\big)+\tilde{C}}{(k+1)^2}
\end{array}
\end{equation}
where $\tilde{C}$ is a finite positive number independent of $k$.
\end{prop}
\noindent Proposition \ref{ACRNL2} is a direct extension of Proposition \ref{ACRNL} and can be obtained by following the roadmap of the proof for the former result.

\section{Performance Evaluation}\label{sec:num-exps}
Next, we present experimental results on both synthetic and real data. Specifically, the first two experiments focus on large-scale network learning with single penalty term to show the efficiency and effectiveness of the proposed algorithms, while the remaining ones assess the accuracy of recovering low rank plus structured sparse transition matrices $B$.

\subsection{Performance metrics and experimental settings}
We introduce the performance metrices used in the numerical work. For network estimation, we use the true positive rate (TPR) and false alarm rate (FAR) defined as:
\begin{itemize}
\item  $\text{TPR} := \frac{\sharp\{\hat{b}_{ij}  \neq0 \ \text{and} \ b_{ij}\neq0\}}{\sharp\{b_{ij}\neq0\}}$
\item  $\text{FAR} := \frac{\sharp\{b_{ij}=0 \ \text{and} \ \hat{b}_{ij} \neq 0\}}{\sharp\{b_{ij}=0\}}$
\end{itemize}
where $b_{ij}$ and $\hat{b}_{ij}$ are the correspnding elemnets in $B$ and $\hat{B}$, respectively.
The estimation error (EE) and out-of-sample prediction error (PE) are defined as
\begin{itemize}
\item $\text{EE} := \frac{\|\hat{B}-B\|_F}{\|B\|_F}$
\item $\text{PE} : = \|\hat{\mathcal{Y}}-\mathcal{Y}\|^2_F/\|\mathcal{Y}\|^2_F$
\end{itemize}

To select the optimal value of the tuning parameters, we combine the three (or two or one, respectively) -dimensional grid search method with the AIC/BIC/forward cross-validation criteria. We will specify the criterion on a case by case basis for the following experiments. In examples B and C, the tuning parameter $\lambda$ is selected by the AIC criterion. A grid of 100 values in the interval $[0, \|\mathcal{X}'\mathcal{Y}\|_{\max}]$ is used for $\lambda$. In examples D, E and F, we utilize a two/three-dimensional grid search to select the optimal values of $\lambda_N$, $\mu_N$ and/or $\nu_N$ as that for $\lambda$. For the experiments employing synthetic data, the tuning parameters are selected by assuming the rank of the true low-rank transition matrix and/or the non-zero group-sparse components of the true group-sparse transition matrix are known. We will specify the forward cross-validation procedure for the real data case in example G.

For all the experiments, the parameters used in the proposed algorithms are depicted in table \ref{table:para1}.
Also we set $\Sigma_{\epsilon} = \epsilon^2I$ and $\epsilon^2 = 1$.  We rescale the entries
of $B$ to ensure stability of the process (the spectral radius $\rho \in (0.45, 0.95)$. All the results are based on
50 replications. Finally, all algorithms are run in the MATLAB R2015a environment on a PC equipped with 12GB memory.
\begin{table}[t!]
\centering 
\begin{tabular}{|c|c|c|c|c|} 
\hline 
Parameters & $\sigma$ & $\eta_{\min}$ & $C$ & $\beta_k$  \\ \hline
Values & 2 & $\|\mathcal{X}'\mathcal{X}\|_2/10$ & 100 & 1/k  \\ [1ex] 
\hline 
\end{tabular}
\caption{Parameter settings in the proposed algorithms for all the experiments.} 
\label{table:para1} 
\end{table} 

\subsection{Large-scale sparse network learning}\label{subsec:sparse}
We start by comparing the performance of the proposed algorithm \ref{alg:fnsl} with FISTA with line search \cite{beck2009fast} to solve problem \eqref{eqn:opt-sparse-lowrank} with a sparse transition matrix.

We consider three different VAR(1) models with $p=800, \, 900$ and $1000$ variables. For each of these models, we generate $N=1000, 1500$, and $2000$ observations from a Gaussian VAR(1) process \eqref{eqn:model-sparse-lowrank}.
The $p \times p$ transition matrix $B$ with sparsity is generated in the following way. First,
the topology is generated from a directed random graph $G(p,\xi)$, where the edge from one node to
another node occurs independently with probability $\xi=10/p$. Then, the strength of the edges is generated
independently from a Gaussian distribution. This process is repeated until we obtain a transition matrix $B$ with
a desired spectral radius $\rho$. We compare TPR, FAR, EE, and computational time (denoted by $T$).

\begin{table}[]
\centering
\begin{tabular}{|c|c|c|c|c|c|}
\hline
p                 & N                 & method  & (TPR, FAR)(\%) & EE & T  \\ \hline
\multirow{6}{*}{800} & \multirow{2}{*}{1000} & FISTA   & (88.5, 14.4) & 0.22 & 11.9  \\ \cline{3-6}
                  &                   & FNSL   & (88.6, 13.9) & 0.22 & $\mathbf{6.1}$   \\ \cline{2-6}
                  & \multirow{2}{*}{1500} & FISTA   &  (90.8, 12.3)   & 0.17   &  14.7    \\ \cline{3-6}
                  &                   & FNSL   &   (90.7, 12.0) &  0.17  & $\mathbf{8.2}$     \\ \cline{2-6}
                  & \multirow{2}{*}{2000} & FISTA   &  (91.3, 11.0)  & 0.14 & 17.6  \\ \cline{3-6}
                  &                   & FNSL   &  (91.3, 11.3) & 0.14 & $\mathbf{10.8}$   \\ \hline
\multirow{6}{*}{900} & \multirow{2}{*}{1000} & FISTA   & (87.8, 15.0)  & 0.24   & 13.1  \\ \cline{3-6}
                  &                   & FNSL   &   (87.8, 14.7)  &  0.24  &  $\mathbf{7.2}$    \\ \cline{2-6}
                  & \multirow{2}{*}{1500} & FISTA   &  (89.1, 10.7)  & 0.19   &  16.3  \\ \cline{3-6}
                  &                   & FNSL   &   (89.1, 10.5) & 0.19 & $\mathbf{8.7}$  \\ \cline{2-6}
                  & \multirow{2}{*}{2000} & FISTA   &  (90.9, 12.0)  & 0.16  & 20.7   \\ \cline{3-6}
                  &                   & FNSL   & (90.9, 11.8)  &  0.16 & $\mathbf{13.5}$  \\ \hline
\multirow{6}{*}{1000} & \multirow{2}{*}{1000} & FISTA   &  (88.5, 15.2)        & 0.25   &  18.6   \\ \cline{3-6}
                  &                   & FNSL   &   (88.4, 14.8) &  0.25  & $\mathbf{11.2}$     \\ \cline{2-6}
                  & \multirow{2}{*}{1500} & FISTA   &  (90.3, 14.1)        &  0.20  &  21.5    \\ \cline{3-6}
                  &                   & FNSL   &   (90.2, 13.8) & 0.20  &  $\mathbf{12.8}$  \\ \cline{2-6}
                  & \multirow{2}{*}{2000} & FISTA   &  (91.2, 12.1)   &  0.16 & 24.7  \\ \cline{3-6}
                  &                   & FNSL   &    (91.2, 12.4)      & 0.16  &  $\mathbf{15.1}$   \\ \hline
\end{tabular}
\caption{Performance comparison of FNSL with a variant of FISTA on large-scale sparse network structure learning problem.}
\label{my-tabel0}
\end{table}

Table \ref{my-tabel0} shows the experimental results for sparse network structure with different network size $p$ and sample size $N$. It can be seen that the proposed algorithm performs similarly to FISTA in terms of TPR, FAR, and estimation error. To show the efficiency of the proposed algorithm, we also compare the computational time in seconds in terms of the convergence of the objective function value. Clearly, the proposed algorithm outperforms FISTA in efficiency, especially when the network and sample size become larger. This is mainly due to the relaxed line search scheme, as previously discussed. To further support our claim, we also show the graphs of the decreasing objective function value vs. CPU, see Figure \ref{fig:fun-value-sparse} when $p=1000$ and $N=2000$.

Table \ref{tbl:Computationcost} shows comparisons between a variant of FISTA and FNSL in terms of objective function value, CPU time in seconds, and the number of matrix products for a network of size $p=1000$ with sample size 1000, 1500 and 2000, respectively. For each data set, FISTA needs around 30 iterations to reach convergence and the total number of line searches is 3 for all iterations, while FNSL needs no more than 20 iterations and the total number of line search is no more than 4 for all iterations. This illustrates the computational savings of the proposed algorithm.

\begin{figure}[tbp]
    \centering
		\includegraphics[width=.8\linewidth]{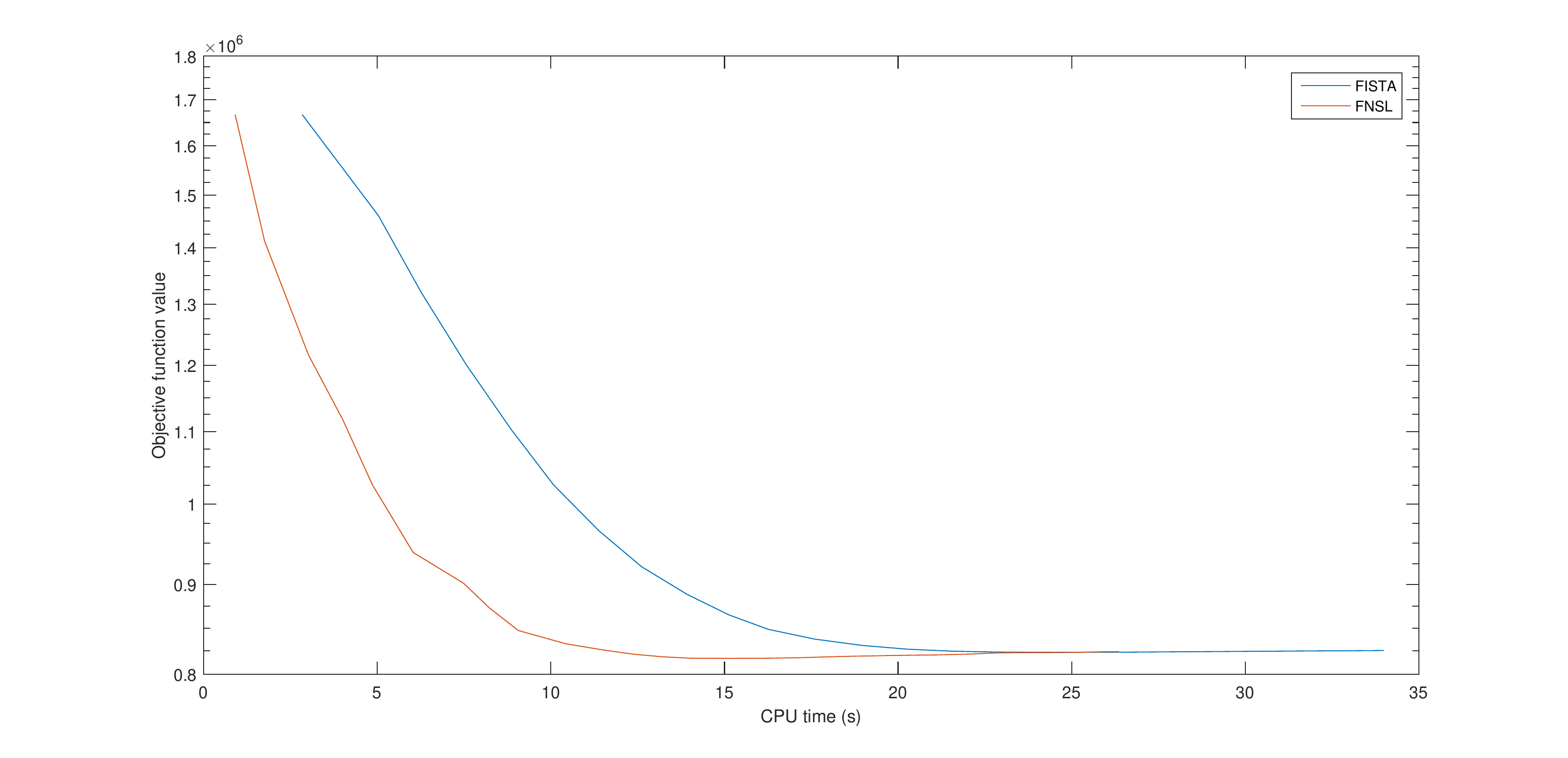}
	\caption{{Performance comparison of a variant of FISTA and the proposed algorithm}: the objective function values vs. CPU time for sparse network learning problem with $p=1000$ and $N=2000$.}
    \label{fig:fun-value-sparse}
\end{figure}


\begin{table}
\centering
 \begin{tabular}{c c c c}
 \hline
 {Algorithms} & {Objective value} & {CPU} & {AX}  \\ [0.5ex]
 \hline
  & $p = 1000, N = 1000$  &  &  \\
 \hline\hline
 FISTA & 4.140e+5 & 18.6 & 66  \\
 FNSL & 4.089e+5 & 11.2 & 44    \\
 \hline
 & $p = 1000, N = 1500$  &  &  \\
 \hline\hline
 FISTA & 6.137e+5 & 21.5 & 62  \\
 FNSL & 6.071e+5 & 12.8 & 38  \\
 \hline
 & $p = 1000, N = 2000$  &  &  \\
 \hline\hline
 FISTA & 8.239e+5 & 24.7 & 66  \\
 FNSL & 8.169e+5 & 15.1 & 41  \\
 \hline
\end{tabular}
\caption{{Comparison of objective function value,  CPU time in seconds, and the number
of matrix products (AX) for a variant of FISTA and FNSL on sparse network structure with different sample size}.}
\label{tbl:Computationcost}
\end{table}
\subsection{Network learning with a low-rank transition matrix}\label{subsec:lowrank}
To further show the efficiency of the proposed algorithms, we compare the performance of a variant of FISTA \cite{ji2009accelerated} and FNSL on estimating low-rank transition matrices.

We consider three different VAR(1) models with $p=200, \, 300$ and $400$ variables. For each of these models, we generate $N=400, 1200$, and $2000$ observations from a Gaussian VAR(1) process \eqref{eqn:model-sparse-lowrank}. The $p \times p$ low-rank transition matrix $B$ is generated with rank $\lfloor p/25 \rfloor+1$. Subsequently, we rescale the entries
of $B$ to ensure the spectral radius $\rho$ lies in $(0.45, 0.95)$. We compare the rank of the estimated transition matrix, denoted by $\hat{r}$, EE, and computational time $T$.

\begin{table}[]
\centering
\begin{tabular}{|c|c|c|c|c|c|}
\hline
p                 & N                 & method & $\hat{r}$ & EE & T \\ \hline
\multirow{6}{*}{200} & \multirow{2}{*}{400} & FISTA    & 8 & 0.80 & 4.9  \\ \cline{3-6}
                  &                   &  FNSL  & 8 & 0.80 & $\mathbf{3.4}$  \\ \cline{2-6}
                  & \multirow{2}{*}{1200} & FISTA   & 8 & 0.63 & 9.7 \\ \cline{3-6}
                  &                   &  FNSL  & 8 & 0.63 &  $\mathbf{6.9}$ \\ \cline{2-6}
                  & \multirow{2}{*}{2000} & FISTA    & 8 & 0.57 & 14.2 \\ \cline{3-6}
                  &                   &  FNSL   & 8 & 0.57 & $\mathbf{10.5}$ \\ \hline
\multirow{6}{*}{300} & \multirow{2}{*}{400} &   FISTA   & 12 & 0.84 &  7.8  \\ \cline{3-6}
                  &                   & FNSL   & 12 & 0.84 &  $\mathbf{5.7}$ \\ \cline{2-6}
                  & \multirow{2}{*}{1200} & FISTA    & 12 & 0.72 & 16.1  \\ \cline{3-6}
                  &                   & FNSL   & 12 & 0.72 & $\mathbf{11.7}$ \\ \cline{2-6}
                  & \multirow{2}{*}{2000} & FISTA    & 12 & 0.68 & 18.2  \\ \cline{3-6}
                  &                   & FNSL    & 12 & 0.68 & $\mathbf{13.8}$  \\ \hline
\multirow{6}{*}{400} & \multirow{2}{*}{400} &   FISTA   & 16 & 0.87 &  8.5  \\ \cline{3-6}
                  &                   & FNSL   & 16 & 0.87 &  $\mathbf{5.9}$ \\ \cline{2-6}
                  & \multirow{2}{*}{1200} & FISTA    & 16 & 0.82 & 20.5  \\ \cline{3-6}
                  &                   & FNSL   & 16 & 0.82 & $\mathbf{17.4}$ \\ \cline{2-6}
                  & \multirow{2}{*}{2000} & FISTA    & 16 & 0.75 &  31.4 \\ \cline{3-6}
                  &                   & FNSL    & 16 & 0.75 & $\mathbf{25.1}$  \\ \hline
\end{tabular}
\caption{{Performance comparison of a variant of FISTA and FNSL on estimation of low-rank transition matrices problems}. }
\label{table:my-labellr}
\end{table} 
Table \ref{table:my-labellr} shows the experimental results for low-rank network structure with different network and sample size. Both FISTA and the proposed algorithm achieve good recovery of the transition matrix $B$ with the correct rank and they have similar performance in terms of estimation error. Clearly, the proposed algorithm outperforms FISTA in efficiency for this case as well. The graphs of the decreasing objective function value vs. CPU are depicted in Figure \ref{fig:fun-value-lr}.
\begin{figure}[tbp]
    \centering
		\includegraphics[width=.8\linewidth]{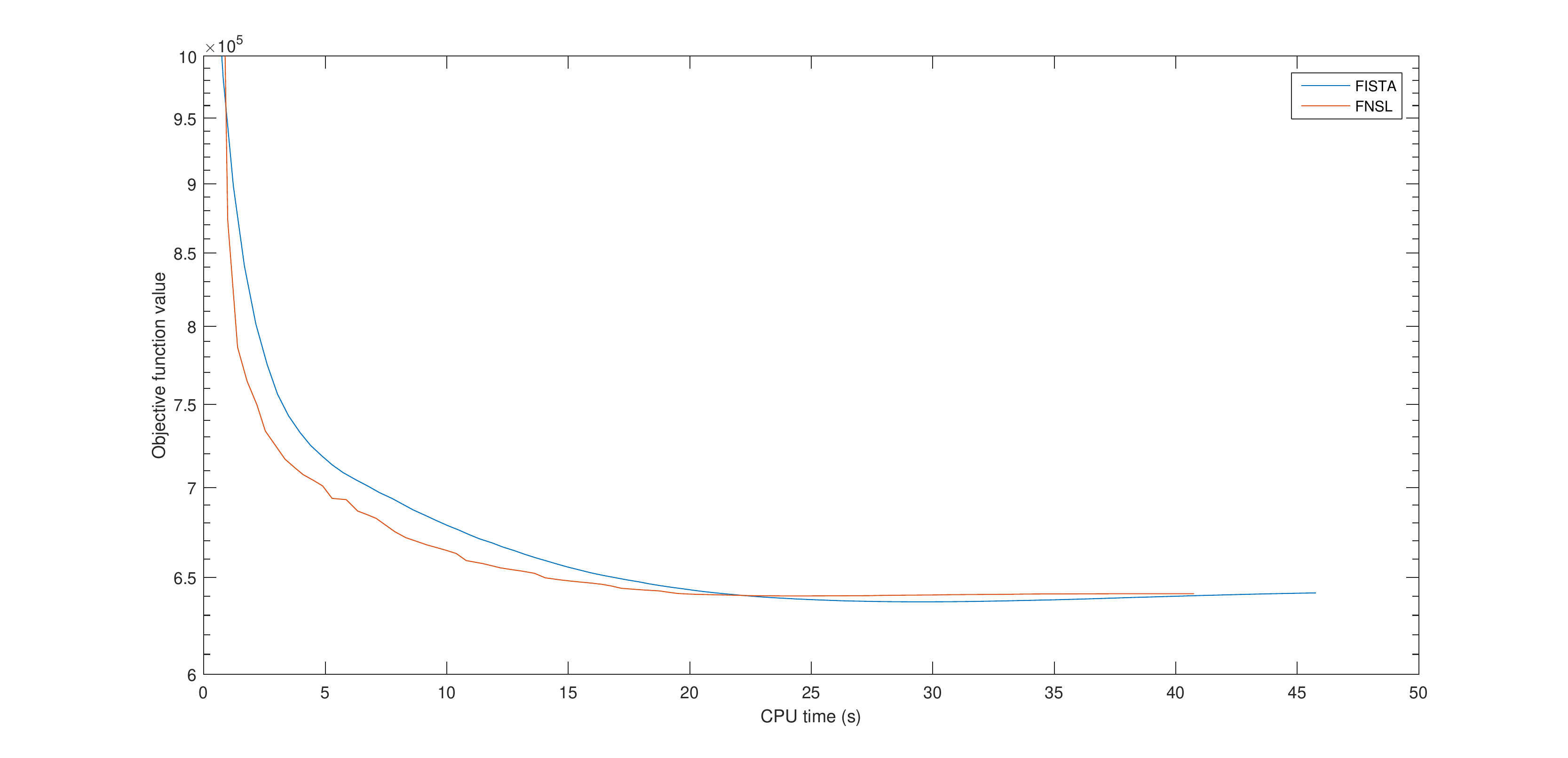}
	\caption{{Performance comparison of a variant of FISTA and the proposed algorithm}: the objective function values vs. CPU time for low-rank
transition matrix estimation problem with with $p=400$ and $N=2000$.}
    \label{fig:fun-value-lr}
\end{figure}

The first two experiments focused on computational efficiency of the proposed algorithms, while retaining good network
estimation properties. Next, we demonstrate their accuracy for learning structured sparse networks.

\subsection{Sparse plus low-rank network learning}\label{sec:splusl}
Next, we investigate estimation of sparse plus low-rank transition matrices and compare it to ordinary least square (OLS) and lasso estimates.

We consider three different VAR(1) models with $p=50, \, 75$ and $100$ variables. For each of these models, we generate $N=100$ and $200$ observations from the model defined in \eqref{eqn:model-sparse-lowrank} where $B$ can be decomposed into a low-rank matrix $L$ of rank $\lfloor p/25 \rfloor+1$ and a sparse matrix $S$ with $2-4\%$ non-zero entries. We rescale the entries of $B$ to ensure stability of the process (the spectral radius is set to $\rho(B)=0.7$). We compare the estimation and out-of-sample prediction errors. The number of out of samples is set to 10.

\begin{table}[t!]
\centering
\begin{tabular}{|c|c|c|}
\hline
$\alpha$ & (TPR, FAR)(\%)  \\ \hline
p/8 & (84.5, 17.4)  \\ \hline
p/4 & (82.4, 17.2)  \\ \hline
p/2 & (80.4, 17.5) \\ \hline
p   & (73.2, 17.1)   \\ \hline
2p  & (54.7, 17.0)   \\ \hline
4p  & (40.2, 17.8)   \\ \hline
8p  & (22.7, 17.4)   \\ \hline
\end{tabular}
\caption{True positive rate and false alarm rate of the L+S model on identifying the sparse component S with different $\alpha$.}
\label{my-label-alpha}
\end{table} 
First,  we study the influence of $\alpha$ in \eqref{eqn:opt-sparse-lowrank} on this learning problem
with $p=50$ and $N=200$. From Table \ref{my-label-alpha}, that a smaller $\alpha$ parameter leads to markedly improved identification of all the true nonzero entries in the sparse component, which consequently leads to better separating the sparse component $S$ from the low-rank component $L$.

\begin{table}[t!]
\centering
\begin{tabular}{|c|c|c|c|c|c|}
\hline
p                 & N                 & model  & (TPR, FAR)(\%) & EE & PE  \\ \hline
\multirow{6}{*}{50} & \multirow{3}{*}{100} & OLS  & (-, -) & 0.84 &  0.72    \\ \cline{3-6}
                  &                   & Lasso & (73.2, 30.0)  & 0.69  & 0.53    \\ \cline{3-6}
                  &                   & L+S  & ($\mathbf{76.3}$, $\mathbf{18.9}$) & $\mathbf{0.48}$  & $\mathbf{0.47}$  \\ \cline{2-6}
                  & \multirow{3}{*}{200} & OLS   & (-, -)  & 0.52  &   0.41   \\ \cline{3-6}
                 &                   & Lasso   & (77.3, 35.0) & 0.57 &  0.45    \\ \cline{3-6}
                  &                   & L+S  & ($\mathbf{80.4}$, $\mathbf{17.5}$)  &  $\mathbf{0.31}$ &  $\mathbf{0.36}$       \\ \hline
\multirow{6}{*}{75} & \multirow{3}{*}{100} & OLS & (-, -)   &  0.75   & 0.37    \\ \cline{3-6}
                 &                   & Lasso & (71.0, 24.7) & 0.75  &  0.37  \\ \cline{3-6}
                  &                   & L+S  & ($\mathbf{79.0}$, $\mathbf{18.0}$)   & $\mathbf{0.51}$  &  $\mathbf{0.29}$   \\ \cline{2-6}
                  & \multirow{3}{*}{200} & OLS  & (-, -) & 0.53    &  0.18  \\ \cline{3-6}
                  &                   & Lasso & (77.0, 28.6)  & 0.67 &  0.22 \\  \cline{3-6}
                  &                   & L+S  & ($\mathbf{83.8}$, $\mathbf{18.3}$)  &  $\mathbf{0.36}$    &    $\mathbf{0.16}$  \\ \hline
\multirow{6}{*}{100} & \multirow{3}{*}{100} & OLS & (-, -) & 3.7  &  4.0 \\ \cline{3-6}
                 &                   & Lasso & (57.3, 29.0) & 1.06   &  1.05   \\ \cline{3-6}
                  &                   & L+S  & ($\mathbf{52.3}$, $\mathbf{20.1}$) &  $\mathbf{0.92}$   & $\mathbf{1.0}$    \\ \cline{2-6}
                  & \multirow{3}{*}{200} & OLS  & (-, -)  & 2.07   &  1.73    \\ \cline{3-6}
                 &                   & Lasso & (59.4, 25.5) &  0.86 &  0.95     \\ \cline{3-6}
                  &                   & L+S & ($\mathbf{60.4}$, $\mathbf{20.5}$)   & $\mathbf{0.72}$  &  $\mathbf{0.90}$   \\ \hline
\end{tabular}
\caption{Performance comparison of L+S with OLS and Lasso.}
\label{my-tabelLS}
\end{table} 
The corresponding estimation errors are reported in Table \ref{my-tabelLS}. In all three settings, we find that the low-rank plus sparse VAR estimates outperform the estimates using ordinary least-squares (OLS) and lasso, as expected. We observe that as the ratio of $N/p$ increases, OLS may produce lower estimation error than lasso, even though the OLS model is not interpretable for this case.  {Also, we observe that the estimation errors of all three methods decrease with increasing sample sizes as expected and predicted by theory. Further, we illustrate how the squared Frobenius norm error in Proposition \ref{prop:main-result-low-rank} of the VAR model with low rank plus sparse transition matrix scales with the sample size $N$ and dimension $p$, when the rank $r$ of $B$ is fixed. The network size $p$ is set to $50, 100, 150$ and $200$, respectively, while the rank $r$ is fixed to be 2 for all $p$. Sparsity $s$ and $\varepsilon$ are defined similarly as above, while the sample size is set to $N\in (150, 5500)$. The squared Frobenius norm error of estimation given by $\|S-\hat{S}\|_{\text{F}}^2+\|L-\hat{L}\|_{\text{F}}^2$ is depicted in the left panel of Figure \ref{fig:ErrorBound-SpL}, which displays the errors for different values of $p$, plotted against the sample
size $N$. As predicted by our theoretical result, the error is larger for larger $p$. The right panel of Figure \ref{fig:ErrorBound-SpL} displays the error against the rescaled sample size $N/(s \, log(p)+rp)$. it can be seen that the corresponding error curves for
different values of $p$ align well, which is consistent with the estimation error obtained in Proposition \ref{prop:main-result-low-rank}}.

\begin{figure}[!t]
    \centering
		\includegraphics[width=.45\linewidth]{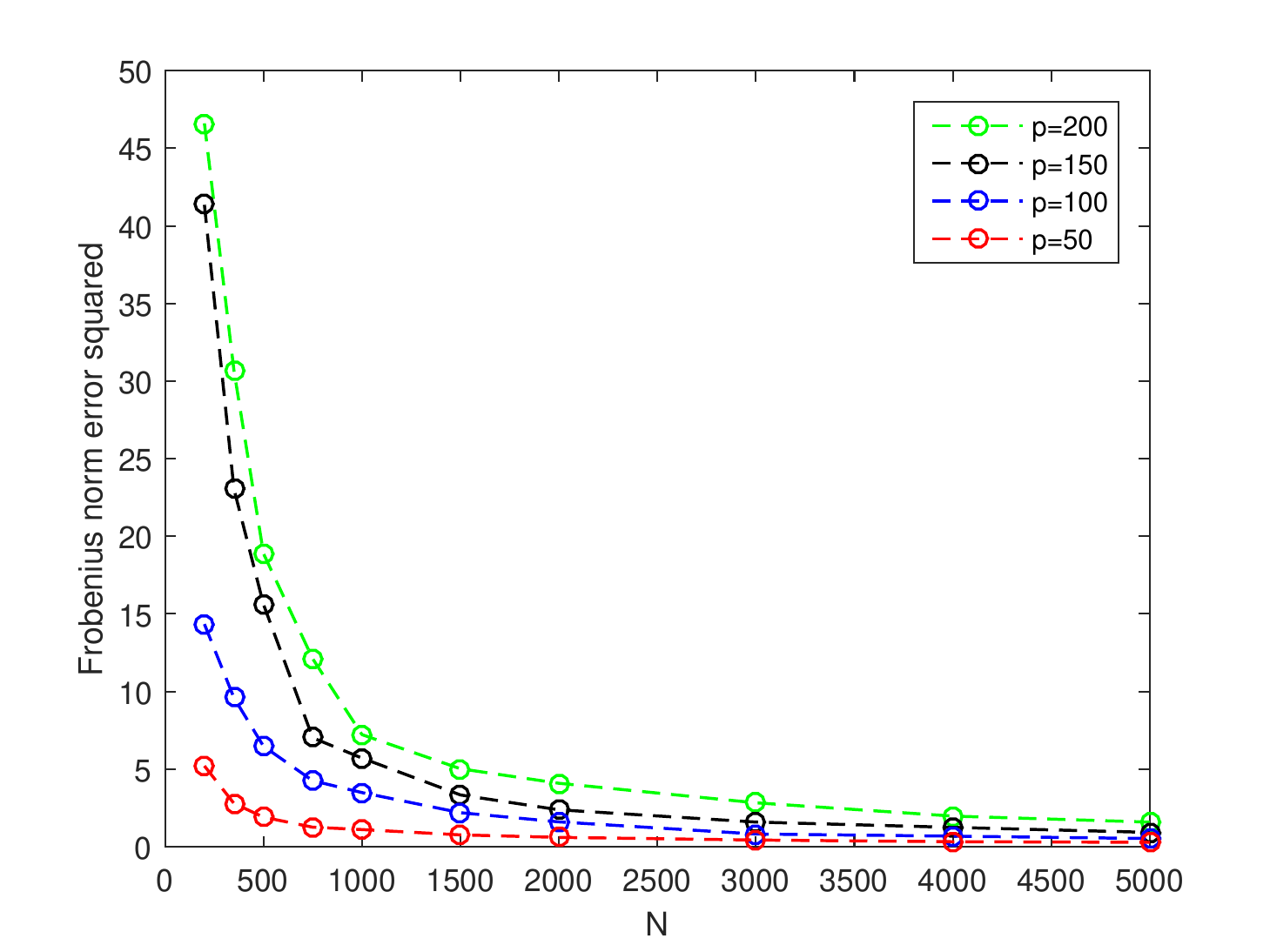}
		\includegraphics[width=.45\linewidth]{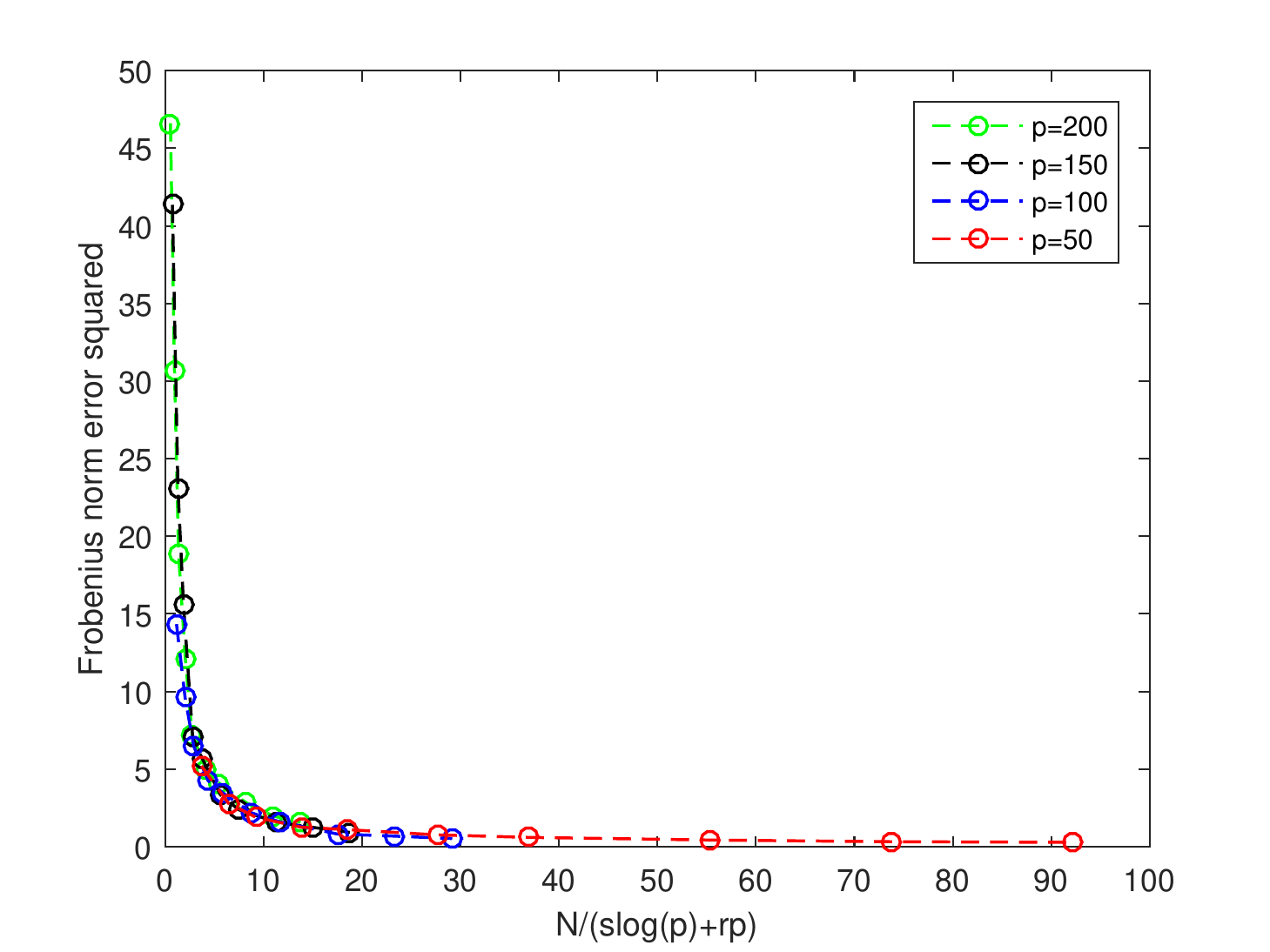}
	\caption{[left]: Estimation error of low rank plus sparse structure $\|S-\hat{S}\|_\text{F}^2 + \|L-\hat{L}\|_\text{F}^2$ with different network size $p$ and sample size $N$. [right]: Estimation error of low rank plus sparse structure $\|S-\hat{S}\|_\text{F}^2 + \|L-\hat{L}\|_\text{F}^2$ with rescaled sample size $N/(slog(p)+rp)$.}
    \label{fig:ErrorBound-SpL}
\end{figure}

In addition to its improved estimation and prediction performance, the low-rank plus
sparse modeling strategy aids in recovering the underlying Granger causal network after accounting
for the latent structure. In Figures \ref{fig:transition-matrix-lr-sparse}, we demonstrate this using a VAR(1) model with
$p = 50$ and $n = 200$. The top panel of the Figures \ref{fig:transition-matrix-lr-sparse} displays the true transition matrix $B$, its low-rank
component $L$ and the structure of its sparse component $S$. The bottom panel of the Figures \ref{fig:transition-matrix-lr-sparse} displays the
structure of the Granger causal networks estimated by the method of Lasso and the low-rank plus sparse
modeling strategy. As predicted by the theory, it can be that the lasso estimate of the Granger causal network selects many false positives due to its failure to account for the latent structure. On the other hand, the sparse component $S$ provides
an estimate exhibiting significantly fewer false positives entries.
\begin{figure}[!t]
      \centering
      {\includegraphics[width=\textwidth]{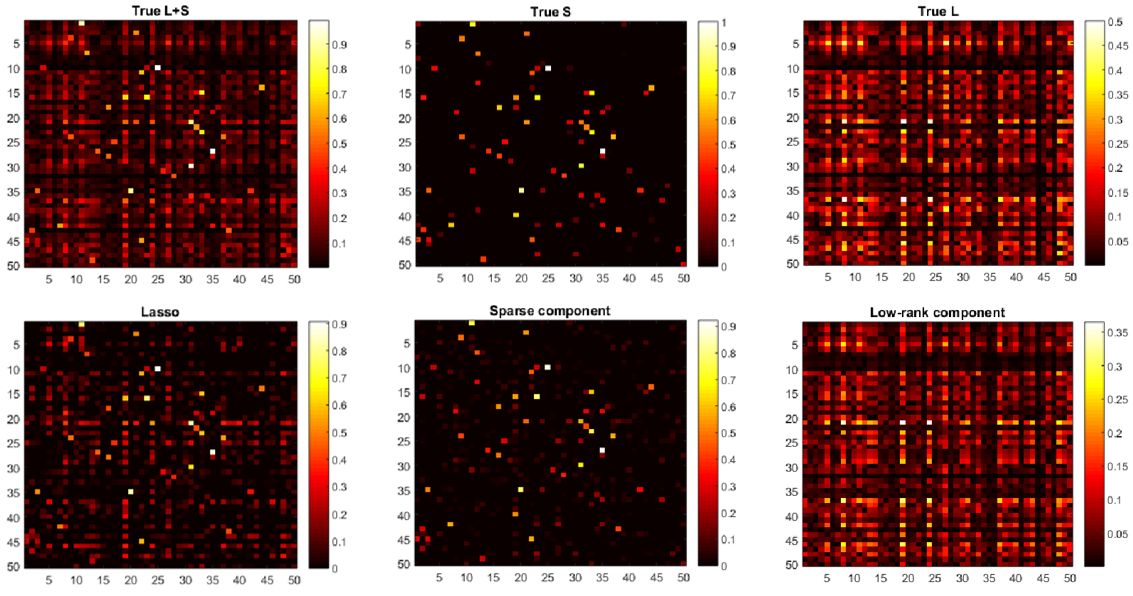}} 
\caption{Estimated Granger causal networks using lasso and low-rank plus sparse VAR estimates. The top panel displays the true transition matrix $B$, the structure of its sparse component $S$ and its low-rank component $L$.  The bottom panel displays the structure of the Granger causal networks estimated by lasso ($\hat{B}_{Lasso}$),  the low-rank plus sparse modeling strategy ($\hat{S}$) and the estimated low-rank component ($\hat{L}$).}
\label{fig:transition-matrix-lr-sparse}
\end{figure}

It is interesting to note that the estimation performance of the regularized estimates in low-rank plus sparse VAR models is worse than the performance of lasso in sparse VAR models of similar dimension \cite{basu2015regularized}, even for the same sample sizes. This is in line with the error bounds presented in Proposition \ref{prop:main-result-low-rank}. The estimation error in low-rank plus sparse models is of the order of $O(rp+s \log p)/N$, while the error  of lasso in sparse VAR models scales at a faster rate of $O(s \log p/N)$. Further note that a $s$-sparse VAR requires estimating $s$ parameters in $S$, while the presence of $r$ factors introduces an additional $rp$  parameters in the loading matrix $\Lambda$.

\subsection{Sparse plus group-sparse plus low-rank network learning}\label{sec:spluslplusg}

Finally, we conduct numerical experiments to assess the performance of
low-rank plus sparse plus group-sparse modeling in VAR analysis and compare it to the performance of sparse plus group-sparse and low-rank plus sparse estimates.

We consider three different VAR(1) models with $p=50, \, 100$ and $150$ variables. For each of these models, we generate $N=200$ and $300$ observations from the Gaussian VAR(1) process defined in \eqref{eqn:model-sparse-lowrank}, where $B$ can be decomposed into a low-rank matrix $L$ of rank $\lfloor p/25 \rfloor+1$, a sparse matrix $S$ with $2-4\%$ non-zero entries, and a group-sparse matrix $G$ with each column corresponding to a different group for a total of $p$ groups. We rescale the entries of $B$ to ensure stability of the process (the spectral radius is set to $\rho(B)=0.7$) and compare the estimation and out-of-sample prediction errors, with the number of out-samples set to 10.

\begin{table}[t!]
\centering
\begin{tabular}{|c|c|c|c|c|c|}
\hline
p                 & N                 & method  &(TPR, FAR)(\%) & EE & PE  \\ \hline
\multirow{6}{*}{50} & \multirow{3}{*}{200} & S+G  & (85.5, 34.1) & 0.46 &  0.53   \\ \cline{3-6}
                  &                   & L+S & (82.6, 26.4)  &  0.41 &  0.51  \\ \cline{3-6}
                  &                   & L+S+G  & ($\mathbf{83.3}$, $\mathbf{26.9}$)  &  $\mathbf{0.41}$ &  $\mathbf{0.51}$  \\ \cline{2-6}
                  & \multirow{3}{*}{300} & S+G  & (91.7, 47.0)  & 0.37 &   0.58   \\ \cline{3-6}
                  &                   & L+S  & (88.6, 24.4)  & 0.31  &  0.56  \\ \cline{3-6}
                  &                   & L+S+G  & ($\mathbf{90.6}$, $\mathbf{24.9}$)  & $\mathbf{0.30}$  & $\mathbf{0.56}$   \\ \hline
\multirow{6}{*}{100} & \multirow{3}{*}{200} & S+G  & (92.3, 49.9)  & 0.48  &    0.73   \\ \cline{3-6}
                  &                   & L+S & (84.3, 28.4) & 0.44   &  0.72    \\ \cline{3-6}
                  &                   & L+S+G & ($\mathbf{85.3}$, $\mathbf{27.3}$)  & $\mathbf{0.44}$   &  $\mathbf{ 0.72}$  \\ \cline{2-6}
                  & \multirow{3}{*}{300} & S+G & (94.8, 49.0)   & 0.44  &   0.72    \\ \cline{3-6}
                  &                   & L+S & (89.6, 25.4) & 0.37   &   0.70   \\ \cline{3-6}
                  &                   & L+S+G  & ($\mathbf{90.0}$, $\mathbf{25.1}$) &  $\mathbf{0.36}$  & $\mathbf{ 0.70}$  \\ \hline
\multirow{6}{*}{150} & \multirow{3}{*}{200} & S+G  & (92.0, 50.5)  & 0.64 &   0.73   \\ \cline{3-6}
                  &                   & L+S & (83.3, 28.8) & 0.57  &   0.71  \\ \cline{3-6}
                  &                   & L+S+G & ($\mathbf{84.0}$, $\mathbf{28.1}$)  & $\mathbf{0.55}$  & $\mathbf{0.70}$   \\ \cline{2-6}
                  & \multirow{3}{*}{300} & S+G & (93.6, 50.2)   & 0.55 &  0.71    \\ \cline{3-6}
                  &                   & L+S &  (85.6, 27.4) &  0.46 &   0.68  \\ \cline{3-6}
                  &                   & L+S+G & ($\mathbf{86.4}$, $\mathbf{27.6}$)   & $\mathbf{0.46}$  &  $\mathbf{0.68}$  \\ \hline
\end{tabular}
\caption{Performance comparison of L+S+G with S+G and L+S.}
\label{my-tabelLSG}
\end{table} 
The corresponding estimation errors are reported in Table \ref{my-tabelLSG}. In those three settings, we find that the low-rank plus sparse plus group-sparse VAR estimates performs only slightly better than low-rank plus sparse VAR estimates. One of the reasons lie in that the ability of the identification will degrade as more structures are involved. The other one is that multiple-times shrinkage for the  multiple structures lead to severe bias estimation. Even though the group structures can be recovered completely, some non-zero elements in sparse component vanished. An ad-hoc way to improve the performance for this case is to combine these two methods together. However,  both methods outperform the estimates using sparse plus group-sparse VAR. We also observe that, as the ratio of $N/p$ increases,  the estimation errors of all three methods decrease with increasing sample sizes as expected and predicted by theory.

\subsection{Structured network learning of asset pricing data}\label{sec:splusl_real_data}
Finally, we employ the proposed framework to learn Granger causal networks of asset pricing data obtained from the University of Chicago's Center for Research in Security Prices (CRSP) and
retrieved from the Wharton Research Data Service (WRDS). Specifically, we examine the network structure of realized volatilities of financial institutions representing banks (BA), primary broker/dealers (PB) and insurance companies (INS).
The analysis was performed across the following time periods: September 2002 - August 2005, September 2006 - August 2008 and September 2010 - August 2012 that correspond to instances before the financial crisis of 2008 (pre-crisis period), the build-up and apex of the crisis and the post-crisis period, respectively. For each period, we collected data on 75 firms with 25 companies in each of the three categories - BA, PB and INS based on the size of the average market capitalization of each firm, but dropped a few due to duplicate/missing observations.
The final form of the variables used are based on the $\log$ transformation of the difference between the highest and lowest stock price during a day that acts as a proxy for realized volatility and subsequently detrending it.

To select the tuning parameters, we employ the following forward cross-validation procedure: (1) We use a time window of length $W$, the available number of time points in the data. Then, starting from time $t$, we use the most recent $W$ observations to estimate $B$ and denote the transition matrix estimate by $\hat{B}_t$. We use the next $W'$ observations (right after $W$ observations) to validate $\hat{B}_t$. (2) We select the optimal tuning parameters
$(\lambda_N^{'}, \mu_N^{'})$ so that
$$
(\lambda_N^{'}, \mu_N^{'}) = \arg \min \left\{\frac{1}{\lfloor \frac{m-500}{25} \rfloor}\sum_{t=500+25*i}^{m}Err(\hat{B}_t)\right\}
$$
where $Err(\hat{B}_t)=\|\mathcal{Y}_{W'}-\mathcal{X}_{W'}\hat{B}_t\|_F^2$ and $i=0, ...$

In our analysis, we set $W=500$ and $W'=50$. Further, to separate the sparse component from the low-rank component as much as possible, we set $\alpha = p/10$. The learned Granger causal network structures (the sparse component $\hat{S}$) estimated by a sparse plus low-rank model are depicted in Figure \ref{fig:transition-matrix-splusl-graph}. It can be seen that even in the presence of a low-rank component, the sparse component exhibits a certain density (about 5\% in
the pre-crisis period, rising to 10\% during the crisis and dropping down to about 6\% in the post-crisis period). This increased connectivity during the crisis period has been observed in Granger causal networks for log-returns
as well \cite{basu2017system}. In the Supplement, we also provide the Granger causal network structures estimated by only assuming sparsity of the transition matrix $B$ (see Supplement Figure 5). A similar increased connectivity pattern
is observed during the crisis period. Also note that after accounting for the low-rank component, the estimated sparse component is significantly more sparse than that estimated by a lasso approach, thus enabling us to better examine
specific firms that are key drivers in the volatility network.

\begin{figure}[!t]
\centering
\begin{subfigure}{.33\textwidth}
  \centering
  \includegraphics[page=2,width=\linewidth]{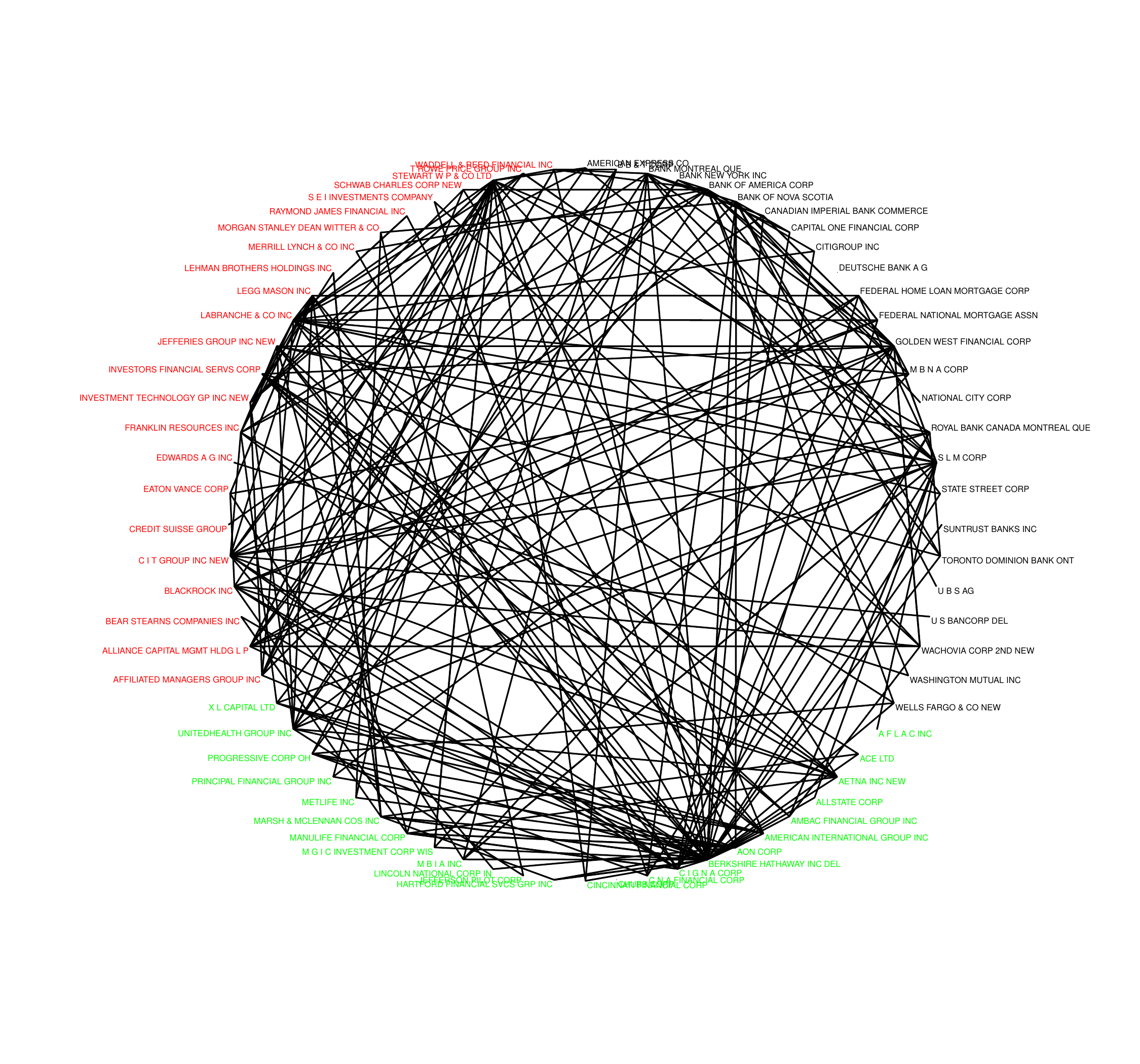}
\end{subfigure}%
\begin{subfigure}{.33\textwidth}
  \centering
  \includegraphics[page=2,width=\linewidth]{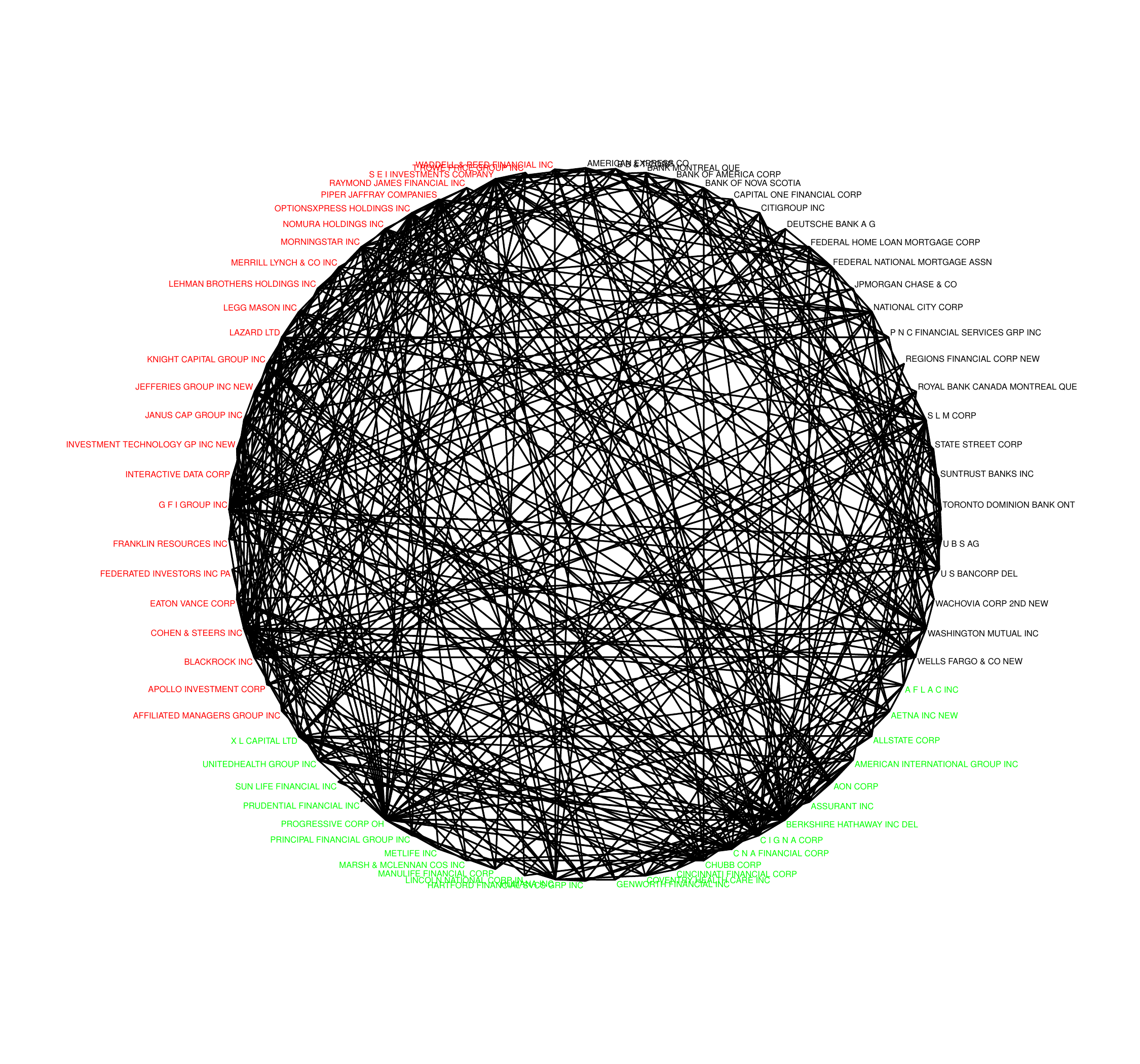}
\end{subfigure}
\begin{subfigure}{.33\textwidth}
  \centering
  \includegraphics[page=2,width=\linewidth]{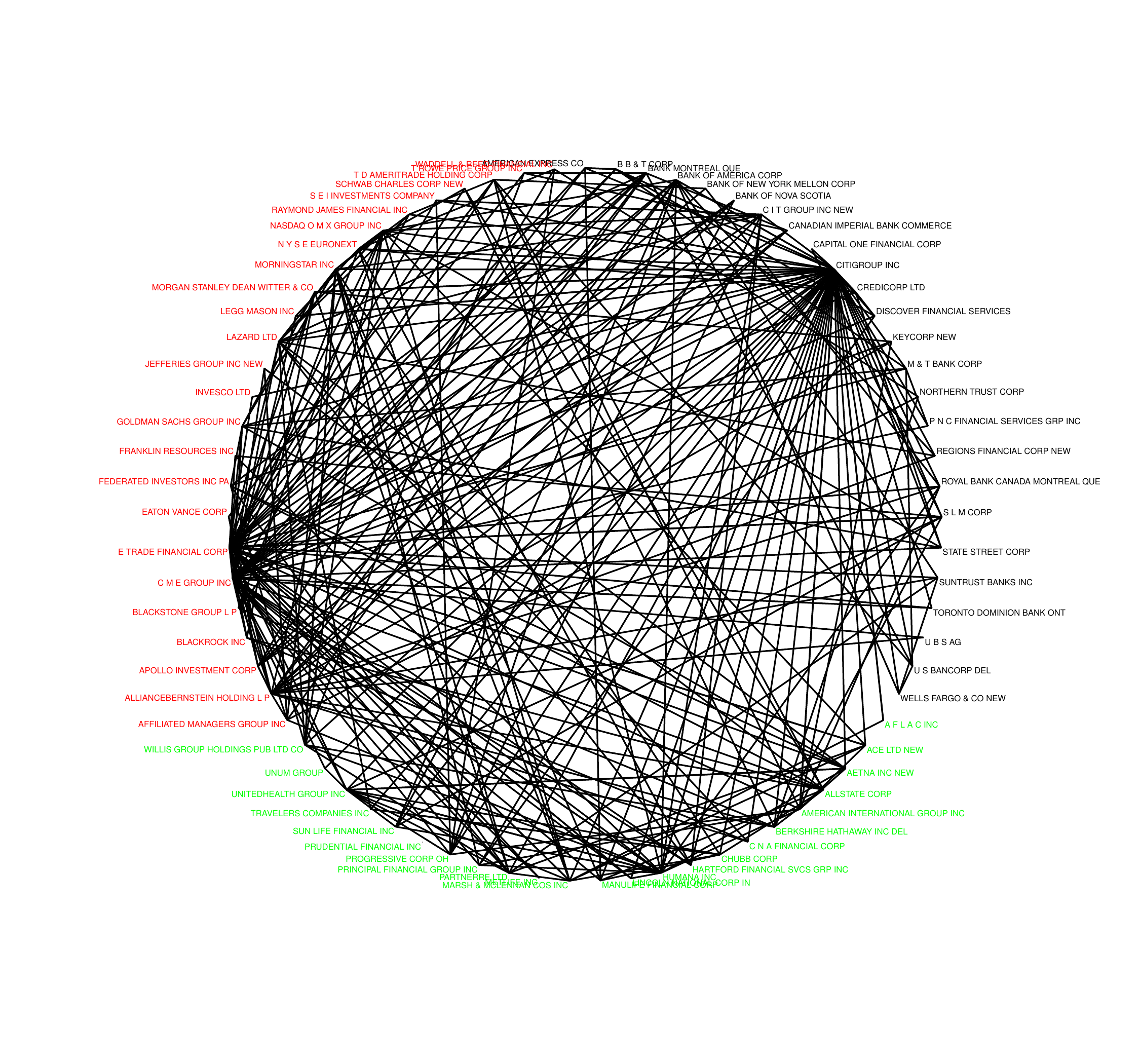}
\end{subfigure}
\caption{Estimated Granger causal networks based on L+S estimates. The left panel depicts the structure of the Granger causal network of the estimated sparse component $\hat{S}$ in the pre-crisis period and has 298 edges, the
middle during the crisis and has 592 edges, while the right panel in the post-crisis periods and has 345 edges.}
\label{fig:transition-matrix-splusl-graph}
\end{figure}

\section{Discussion}

Our modeling and technical developments were based on a VAR(1) model. However, it can be generalized to VAR(d) models in different ways, depending on the context of the application. One possible formulation where the low rank component stays the same across lags can be expressed as
\begin{equation*}
    X_t = \sum_{\ell=1}^d (L+S_\ell) X_{t-\ell} + v_t.
\end{equation*}
This model can be posed as a 1-order (L+S) model on the concatenated process \\
$ \tilde{X}_t = \left[X_t^\top, X_{t-1}^\top, \ldots, X_{t-d+1}^\top \right]^\top$ using the standard transformation \cite{lutkepohl2005new}:
\begin{eqnarray*}
    \underbrace{\left[
    \begin{array}{c}
    X_t \\ X_{t-1} \\ \vdots \\ X_{t-d}
    \end{array}
    \right]
    }_{\tilde{X}_t}
    =
    \underbrace{\left[
    \begin{array}{cccc}
    L & L & \ldots & L \\
    0 & 0 & \ldots & 0 \\
    \vdots & \vdots & \ddots & \vdots \\
    0 & 0 & \ldots & 0
    \end{array}
    \right]
    }_{\tilde{L}}
    \underbrace{\left[
    \begin{array}{c}
    X_{t-1} \\ X_{t-2} \\ \vdots \\ X_{t-d-1}
    \end{array}
    \right]
    }_{\tilde{X}_{t-1}}
    +
    \underbrace{\left[
    \begin{array}{cccc}
    S_1 & S_2 & \ldots & S_d \\
    I & 0 & \ldots & 0 \\
    \vdots & \vdots & \ddots & \vdots \\
    0 & \ldots & I & 0
    \end{array}
    \right]
    }_{\tilde{S}}
    \underbrace{\left[
    \begin{array}{c}
    X_{t-1} \\ X_{t-2} \\ \vdots \\ X_{t-d-1}
    \end{array}
    \right]
    }_{\tilde{X}_{t-1}}
    +
    \underbrace{\left[
    \begin{array}{c}
    v_t  \\
    0 \\
    \vdots \\
    0
    \end{array}
    \right]
    }_{\tilde{v}_t}.
\end{eqnarray*}

Since the matrices $\tilde{L}$ and $\tilde{S}$ are respectively low-rank and sparse, our proposed algorithm can be used and the error bounds will be applicable. Further, to maintain the special structure of these new matrices, additional constraints can be imposed to the posited objective function or the final estimates can be projected to this space.

Finally, motivated by \cite{zorzi2016ar, zorzi2017sparse}, it would be of interest to consider analogous developments in the frequency domain and address identifiability issues, as well as establish finite sample error bounds.

\section*{Acknowledgment}

This work was supported by a UF Informatics Institute Fellowship to XL, NSF grant DMS-1812128 to SB,
and NSF grants IIS 1632730, CCF 1540093, DMS 1545277 to GM.

\appendix

\section{Estimation Error Bounds}
The proof of Proposition \ref{CLplusS} and part (a) of Proposition \ref{CSplusGS} can be easily obtained by the following proof for
part (b) of Corollary \ref{CSplusGS}.

\begin{proof}[Proof of part (b) of Corollary \ref{CSplusGS}]
By the optimality of $(\hat{L}, \hat{S}, \hat{G})$ and the feasibility of $(L^*, S^*, G^*)$, we have
\begin{equation}\label{ee1}
\begin{array}{lr}
\frac{1}{2}\|\mathcal{Y}-\mathcal{X}(\hat{L}+\hat{S}+\hat{G})\|_{F}^2+\lambda_N\|\hat{L}\|_{\ast}+\mu_N\|\hat{S}\|_1\\+\nu_N\|\hat{G}\|_{2,1}
\leq \frac{1}{2}\|\mathcal{Y}-\mathcal{X}(L^*+S^*+G^*)\|_{F}^2\\+\lambda_1\|L^*\|_{\ast}+\mu_N\|S^*\|_1+\nu_N\|G^*\|_{2,1}
\end{array}
\end{equation}
By setting $\hat{\Delta}^{L}=\hat{L}-L^*$, $\hat{\Delta}^{S}=\hat{S}-S^*$, and $\hat{\Delta}^{G}=\hat{G}-G^*$ and combining with $\mathcal{Y} = \mathcal{X}(L^*+S^*+G^*)+E$, we have
$$
\begin{array}{lr}
\frac{1}{2}\|\mathcal{X}(\hat{\Delta}^{L}+\hat{\Delta}^{S}+\hat{\Delta}^{G})\|_{F}^2 \leq \langle \hat{\Delta}^{L}+\hat{\Delta}^{S}+\hat{\Delta}^{G}, \mathcal{X}^{'}E\rangle +\lambda_N\|L^*\|_{\ast}+\mu_N\|S^*\|_1+\nu_N\|G^*\|_{2,1}\\
-\lambda_N\|L+\hat{\Delta}^{L}\|_{\ast}-\mu_N\|S^*+\hat{\Delta}^{S}\|_1-\nu_N\|G^*+\hat{\Delta}^{G}\|_{2,1}
\end{array}
$$
By Lemma 1 in \cite{agarwal2012} and lemma 2.3 in \cite{recht2010guaranteed}, we obtain
\begin{equation}\label{ee2}
\begin{array}{lr}
\frac{1}{2}\|\mathcal{X}(\hat{\Delta}^{L}+\hat{\Delta}^{S}+\hat{\Delta}^{G})\|_{F}^2 \leq \langle \hat{\Delta}^{L}+\hat{\Delta}^{S}+\hat{\Delta}^{G},  \mathcal{X}^{'}E\rangle+\lambda_N(\|\hat{\Delta}^{L}_{A}\|_{\ast}-\|\hat{\Delta}^{L}_{B}\|_{\ast})\\+2\lambda_N\sum_{j=r+1}^{d}\sigma_{j}(L^*)+
\mu_N(\|\hat{\Delta}^{S}_{M}\|_1-\|\hat{\Delta}^{S}_{M^{\bot}}\|_1)+2\mu_N\|S^*_{M^{\bot}}\|_1
\\+\nu_N(\|\hat{\Delta}^{G}_{N}\|_{2,1}-\|\hat{\Delta}^{G}_{N^{\bot}}\|_{2,1})+2\nu_N\|G^*_{N^{\bot}}\|_{2,1}
\end{array}
\end{equation}
where the matrices $(A, B) \in \left\{(\bar{A},\bar{B}): \bar{A}\bar{B}'=0 ~~ \& ~~ \bar{A}'\bar{B}=0 \right\}$, $(M, M^{\bot})$ and $(N, N^{\bot})$ denote an arbitrary subspace pair for which $\|S\|_1$ and $\|G\|_{2,1}$ are decomposable, respectively.
Since
\begin{equation}\label{ee3}
\begin{array}{lr}
\langle \hat{\Delta}^{L}+\hat{\Delta}^{S}+\hat{\Delta}^{G}, \mathcal{X}^{'}E\rangle\leq \|\hat{\Delta}^{L}\|_{\ast}\|\mathcal{X}^{'}E\|_{2}+\|\hat{\Delta}^{S}\|_1\|\mathcal{X}^{'}E\|_{\max}
+\|\hat{\Delta}^{G}\|_{2,1}\|\mathcal{X}^{'}E\|_{2,\max} \\ \leq  (\|\hat{\Delta}^{L}_{A}\|_{\ast}+\|\hat{\Delta}^{L}_{B}\|_{\ast})\|\mathcal{X}^{'}E\|_{2}+
(\|\hat{\Delta}^{S}_{M}\|_1+\|\hat{\Delta}^{S}_{M^{\bot}}\|_1) \|\mathcal{X}^{'}E\|_{\max}
+(\|\hat{\Delta}^{G}_{N}\|_{2,1}\\+\|\hat{\Delta}^{G}_{N^{\bot}}\|_{2,1})\|\mathcal{X}^{'}E\|_{2,\max}
\end{array}
\end{equation}
Substituting \eqref{ee3} into \eqref{ee2} and recalling conditions for $\lambda_N$, $\mu_N$ and $\nu_N$, we have
\begin{equation}\label{ee30}
\begin{array}{lr}
\frac{1}{2}\|\mathcal{X}(\hat{\Delta}^{L}+\hat{\Delta}^{S}+\hat{\Delta}^{G})\|_{F}^2 \leq \frac{3}{2}\lambda_N\|\hat{\Delta}^{L}_{A}\|_{\ast}+\frac{3}{2}\mu_N\|\hat{\Delta}^{S}_{M}\|_1+
+\frac{3}{2}\nu_N\|\hat{\Delta}^{G}_{N}\|_{2,1}+2\lambda_N\sum_{j=r+1}^{d} \sigma_{j}(L^*)
\\+2\mu_N\|S^*_{M^{\bot}}\|_1+2\nu_N\|G^*_{N^{\bot}}\|_{2,1}
\end{array}
\end{equation}
By the RSC condition, the constraints on $L$ and $G$ in \eqref{eqn:opt-sparse-lowrank}, and the definition of $\mu_N$ and $\nu_N$, we have
$$
\begin{array}{lr}
\frac{1}{2}\|\mathcal{X}(\hat{\Delta}^{L}+\hat{\Delta}^{S}+\hat{\Delta}^{G})\|_{F}^2 \geq
\frac{\zeta}{2}\|\hat{\Delta}^{L}+\hat{\Delta}^{S}+\hat{\Delta}^{G}\|_{F}^2 \\
\geq \frac{\zeta}{2}\|\hat{\Delta}^{L}\|_{F}^2+\frac{\zeta}{2}\|\hat{\Delta}^{S}\|_{F}^2 +\frac{\zeta}{2}\|\hat{\Delta}^{G}\|_{F}^2 -\zeta|\langle \hat{\Delta}^{L}, \hat{\Delta}^{S} \rangle|\\ ~~~~-\zeta|\langle \hat{\Delta}^{L}, \hat{\Delta}^{G} \rangle|-\zeta|\langle \hat{\Delta}^{G}, \hat{\Delta}^{S} \rangle| \\
\geq \frac{\zeta}{2}\|\hat{\Delta}^{L}\|_{F}^2+\frac{\zeta}{2}\|\hat{\Delta}^{S}\|_{F}^2 +\frac{\zeta}{2}\|\hat{\Delta}^{G}\|_{F}^2-\zeta\|\hat{\Delta}^{L}\|_{\max}
\|\hat{\Delta}^{S}\|_{1}\\-\zeta\|\hat{\Delta}^{L}\|_{2,\max}\|\hat{\Delta}^{G}\|_{2,1}
~~~~ -\zeta\|\hat{\Delta}^{G}\|_{\max}\|\hat{\Delta}^{S}\|_{1}\\
\geq \frac{\zeta}{2}\|\hat{\Delta}^{L}\|_{F}^2+\frac{\zeta}{2}\|\hat{\Delta}^{S}\|_{F}^2 +\frac{\zeta}{2}\|\hat{\Delta}^{G}\|_{F}^2-\frac{\mu_N}{2}\|\hat{\Delta}^{S}\|_{1}\\ ~~~~-\frac{\nu_N}{2}\|\hat{\Delta}^{G}\|_{2,1}-
\frac{\mu_N}{2}\|\hat{\Delta}^{S}\|_{1}
\end{array}
$$
Inserting the above inequality into \eqref{ee30}, we have
$$
\begin{array}{lr}
\frac{\zeta}{2}(\|\hat{\Delta}^{L}\|_{F}^2+\|\hat{\Delta}^{S}\|_{F}^2 +\|\hat{\Delta}^{G}\|_{F}^2) \leq \frac{3}{2}\lambda_N\|\hat{\Delta}^{L}_{A}\|_{\ast}\\+\frac{3}{2}\mu_N\|\hat{\Delta}^{S}_{M}\|_1
+\frac{3}{2}\nu_N\|\hat{\Delta}^{G}_{N}\|_{2,1}+\mu_N\|\hat{\Delta}^{S}\|_{1}
+\frac{\nu_N}{2}\|\hat{\Delta}^{G}\|_{2,1}\\
+2\lambda_N\sum_{j=r+1}^{d}\sigma_{j}(L^*)+2\lambda_N\|S^*_{M^{\bot}}\|_1+2\mu_N\|G^*_{N^{\bot}}\|_{2,1}
\end{array}
$$
By the compatibility constant in \cite{agarwal2012},  we have
$$
\begin{array}{lr}
\frac{\zeta}{2}(\|\hat{\Delta}^{L}\|_{\text{F}}^2+\|\hat{\Delta}^{S}\|_{\text{F}}^2 +\|\hat{\Delta}^{G}\|_{\text{F}}^2) \leq (\frac{3}{2}\lambda_N\sqrt{2r})\|\hat{\Delta}^{L}\|_{\text{F}}+(\frac{5}{2}\mu_N)\sqrt{s}\|\hat{\Delta}^{S}\|_{\text{F}}
+2\nu_N\sqrt{g}\|\hat{\Delta}^{G}\|_{\text{F}}
\\+2\lambda_N\sum_{j=r+1}^{d}\sigma_{j}(L^*)+2\mu_N\|S^*_{M^{\bot}}\|_1+2\nu_N\|G^*_{N^{\bot}}\|_{2,1}
\end{array}
$$
By our assumptions, we have
$$
\begin{array}{lr}
\frac{\zeta}{4}(\|\hat{\Delta}^{L}\|_{F}^2+\|\hat{\Delta}^{S}\|_{F}^2 +\|\hat{\Delta}^{G}\|_{F}^2)
\leq\sqrt{(\frac{3}{2}\lambda_1\sqrt{2r})^2+(\frac{5}{2}\lambda_2\sqrt{s})^2+(2\lambda_3\sqrt{g})^2}
\sqrt{\|\hat{\Delta}^{L}\|_{F}^2+\|\hat{\Delta}^{S}\|_{F}^2+\|\hat{\Delta}^{G}\|_{F}^2}
\end{array}
$$
Combining with the inequality  $\|\hat{\Delta}^{L}\|_{F}^2+\|\hat{\Delta}^{S}+\hat{\Delta}^{G}\|_{F}^2
\leq 2(\|\hat{\Delta}^{L}\|_{F}^2+\|\hat{\Delta}^{S}\|_{F}^2 +\|\hat{\Delta}^{G}\|_{F}^2)$, we conclude part (b) of Corollary \ref{CSplusGS}.
\end{proof}

\section{Deviation Bounds}
\begin{proof}[Proof of Proposition \ref{prop:conc-sparse-lowrank}]
\begin{enumerate}
\item We want to find upper bounds on $\| \mathcal{X}'E/N\|_{\max}$, $\| \mathcal{X}'E/N\|$ and $\|\mathcal{X}'E/N \|_{2,\max}$ that hold with high probability. Note that such an upper bound for $\|\mathcal{X}'E/N\|_{\max}$ has already been derived in \cite{basu2015regularized}. Here we adopt a different technique that takes a unified approach to provide upper bounds on both  quantities. To this end, note that the two norms have the following representations
\begin{equation*}
\begin{array}{lr}
\frac{1}{N}\|\mathcal{X}'E\| = \sup_{u, v \in \mathbb{S}^{p-1}} \frac{1}{N} u'\mathcal{X}'Ev , \\ 
\frac{1}{N}\|\mathcal{X}'E\|_{\max} = \sup_{u, v \in \{e_1, \ldots, e_p\}} \frac{1}{N} u'\mathcal{X}'Ev
\end{array}
\end{equation*}
For any given $u,v \in \mathbb{S}^{p-1}$, we first provide a bound on $u'(\mathcal{X}'E/N)v$.

Using Proposition 2.3 of \cite{basu2015regularized}, we obtain
\begin{equation}\label{eqn:single-dev-bound}
\mathbb{P}\left[|u'(\mathcal{X}'E/N)v| > 2\pi \eta \phi(B, \Sigma_{\epsilon}) \right] \le 6 \exp\left[-cN \min \{\eta, \eta^2 \} \right]
\end{equation}
for any $u,v \in \mathbb{S}^{p-1}$ and any $\eta > 0$.

To derive the deviation bound on $\|\mathcal{X}'E/N \|_{\max}$, we simply take a union bound over the $p^2$ possible choices of $u, v \in \{e_1, e_2, \ldots, e_p \}$. This leads to
\begin{equation*}
\begin{array}{lr}
\mathbb{P}\left[\|\mathcal{X}'E/N \|_{\max} > 2\pi \eta \phi(B, \Sigma_{\epsilon}) \right]  \le 6 \exp\left[-cN \min \{\eta, \eta^2 \} +2 \log p\right]
\end{array}
\end{equation*}
Since $N \succsim p$, we can set $\eta = \sqrt{(2+c_1) \log p / cN}$ so that $\eta < 1$ (i.e., $\eta^2 < \eta$) will be satisfied for large enough $N$. This implies that
\begin{equation*}
\mathbb{P}\left[\|\mathcal{X}'E/N \|_{\max} > c_0 \phi(B, \Sigma_{\epsilon}) \right] \le c_1 \exp \left[-c_2 \log p\right]
\end{equation*}
for some universal constants $c_i > 0$.

{Similarly, for any group $G_i$ of size $m_i$, we have
\begin{equation}\label{eqn:single-grp-dev-bound}
\begin{array}{lr}
\mathbb{P}\left[ \left\|vec(\mathcal{X}_r'E_s/N, ~~ (r,s) \in G_i) \right\| > 2\pi \sqrt{m_i}  \eta \phi(B, \Sigma_{\epsilon}) \right]  \le 6 \exp\left[-cN \min \{\eta, \eta^2 \} + \log m_i \right].
\end{array}
\end{equation}
Taking a union bound over $K$ non-overlapping groups $G_i$ leads to
\begin{equation}\label{eqn:multiple-grp-dev-bound}
\begin{array}{lr}
\mathbb{P}\left[ \left\|\mathcal{X}'E/N \right\|_{2, max} > 2\pi \sqrt{m} \eta \phi(B, \Sigma_{\epsilon}) \right]  \le 6 \exp\left[-cN \min \{\eta, \eta^2 \} + \log p \right],
\end{array}
\end{equation}
where $m = \max_{i=1, \ldots, K} m_i$. As before, setting $\eta =
\sqrt{\log p /N}$ implies
\begin{equation}\label{eqn:max-grp-dev-bound}
\begin{array}{lr}
\mathbb{P}\left[ \left\|\mathcal{X}'E/N \right\|_{2,max} > 2\pi \sqrt{m \log p / N}  \phi(B, \Sigma_{\epsilon}) \right]  \le c_1 \exp\left[-c_2\log p \right]
\end{array}
\end{equation}
for some $c_i > 0$.
}

To derive the deviation bound on the spectral norm, we discretize the unit ball $S^{p-1}$ using an $\epsilon$-net $\mathcal{N}$ of cardinality at most $(1+2/\epsilon)^p$. An argument along the line of Supplementary Lemma F.2 in \cite{basu2015regularized} then shows that for a small enough $\epsilon>0$,
\begin{equation*}
\sup_{u,v \in \mathbb{S}^{p-1}} |u'(\mathcal{X}'E/N)v| \le
\kappa \sup_{u,v \in \mathcal{N}} |u'(\mathcal{X}'E/N)v|
\end{equation*}
for some constant $\kappa > 1$, possibly dependent on $\epsilon$. As before, taking a union bound over the $(1+2/\epsilon)^{2p}$ choices of $u$ and $v$, we get
\begin{equation*}
\begin{array}{lr}
\mathbb{P}\left[\|\mathcal{X}'E/N \| > 2\pi \kappa \eta \phi(B, \Sigma_{\epsilon}) \right] \le 6 \exp\left[-cN \min \{\eta, \eta^2 \} +2p \log(1+2/\epsilon)\right]
\end{array}
\end{equation*}
Since $N \succsim p$, choosing $\eta = \sqrt{(c_1+2\log(1+2/\epsilon))p/cN}$ ensures $\eta < 1$ for large enough $N$. Setting $\eta$ as above concludes the proof.

\item We want to obtain a lower bound on the  minimum eigenvalue of $\mathcal{X}'\mathcal{X}/N$ that holds with high probability.

Since $\Lambda_{\min}\left(\mathcal{X}'\mathcal{X}/N \right) = \inf_{v \in \mathbb{S}^{p-1}} v'(\mathcal{X}'\mathcal{X}/N)v$, we start with the single deviation bound of Proposition 2.3 in \cite{basu2015regularized},
\begin{equation*}
\begin{array}{lr}
\mathbb{P} \Big[\left|v'\left(\mathcal{X}'\mathcal{X}/N - \Gamma_X(0)\right)v\right| > 2\pi \eta \mathcal{M}(f_X) \Big]  \le 2\exp \left[ -cN \min\{ \eta, \eta^2\} \right]
\end{array}
\end{equation*}
for any $v \in \mathbb{S}^{p-1}$ and $\eta > 0$.

The next step is to extend this single deviation bound uniformly on the set $\mathbb{S}^{p-1}$. As in the proof of part 1, we construct a $\epsilon$-net of cardinality at most $(1+2/\epsilon)^p$ and approximate the quadratic form using its values on the net. This yields the following deviation bound
\begin{equation*}
\begin{array}{lr}
\mathbb{P} \Big[\sup_{v \in \mathbb{S}^{p-1}}\left|v'\left(\frac{\mathcal{X}'\mathcal{X}}{N} - \Gamma_X(0)\right)v\right| > 2\kappa \pi \eta \mathcal{M}(f_X) \Big]  \le 2\exp \left[ -cN \min\{ \eta, \eta^2\} + p \log\left(1+\frac{2}{\epsilon}\right) \right]
\end{array}
\end{equation*}
for some constant $\kappa >1$. Seting $\eta = \EuFrak{m}(f_X)/4 \kappa \pi \mathcal{M}(f_X) < 1$ and noting that $N \succsim  \mathcal{M}^2(f_X)/ \EuFrak{m}^2(f_X) p$, we conclude
\begin{equation*}
\begin{array}{lr}
\mathbb{P} \Big[\sup_{v \in \mathbb{S}^{p-1}}\left|v'\left(\mathcal{X}'\mathcal{X}/N - \Gamma_X(0)\right)v\right| > \EuFrak{m}(f_X)/2 \Big]  \le c_0\exp \left[ -c_1 \log p \right]
\end{array}
\end{equation*}
The result follows from  the lower bound on $\EuFrak{m}(f_X)$ presented in \eqref{eqn:bound-stability-measures} and the fact that $v'\Gamma_X(0)v \ge \EuFrak{m}(f_X)$ for all $v \in \mathbb{S}^{p-1}$.
\end{enumerate}
\end{proof}
\begin{proof}[Proof of Proposition \ref{prop:main-result-low-rank} and \ref{prop:main-result-low-rank+G}]
Clearly, setting $\zeta$ to the lower bound on $\EuFrak{m}(f_X)$ as in \eqref{eqn:bound-stability-measures} satisfies the RSC.
Combining the estimates of Proposition \ref{CLplusS} and \ref{prop:conc-sparse-lowrank} leads to proposition \ref{prop:main-result-low-rank} and \ref{prop:main-result-low-rank+G} after simple algebraic computation.
\end{proof}

\section{Convergence Analysis}
In the following proof, we denote $\frac{1}{2} \left\| \mathcal{Y} - \mathcal{X}B \right\|^2_F$  and the regularization term by $H(B)$ and $P_B(B,\lambda)$, respectively.
\begin{proof}[Proof of Proposition \ref{ACRNL}]
By the differentiability of $H$, we have
\begin{equation}\label{eq25}
H(B_{i+1}^{ag})=H(B_{i}^{md})+\int_0^1<\nabla H(B_{i}^{md}
 +\tau(B_{i+1}^{ag}-B_{i}^{md})), B_{i+1}^{ag}-B_{i}^{md}>d\tau
\end{equation}
Then,  by the definition of $H(B)$ and $B_{i+1}^{ag}$, and the relationship $B_{i+1}^{ag}-B_{i}^{md}=\alpha_i(B_{i+1}-B_i)$, we have
\begin{equation}\label{eq26a}
\begin{split}
l(B_{i+1}^{ag}) &=H(B_{i+1}^{ag})+P_B(B_{i+1}^{ag},\lambda)\\
&=H(B_{i}^{md})+\int_0^1\langle \mathcal{X}^{T}\mathcal{X}(B_{i}^{md}+\tau(B_{i+1}^{ag}-B_{i}^{md}))
-\mathcal{X}^{T}\mathcal{Y}, B_{i+1}^{ag}-B_{i}^{md}\rangle d\tau+ P_B(B_{i+1}^{ag},\lambda)\\
&=H(B_{i}^{md})+\int_0^1\langle \mathcal{X}^{T}(\mathcal{X}B_{i}^{md}-\mathcal{Y}), B_{i+1}^{ag}-B_{i}^{md}\rangle d\tau
 +\int_0^1\tau\|\mathcal{X}(B_{i+1}^{ag}-B_{i}^{md})\|^2d\tau+ P_B(B_{i+1}^{ag},\lambda)\\
&=H(B_{i}^{md})+\langle\nabla H(B_{i}^{md}), B_{i+1}^{ag}-B_{i}^{md}\rangle
 +\frac{1}{2}\|\mathcal{X}(B_{i+1}^{ag}-B_{i}^{md})\|^2+ P_B(B_{i+1}^{ag},\lambda)\\
&=H(B_{i}^{md})+(1-\alpha_i)\langle\nabla H(B_{i}^{md}), B_{i}^{ag}-B_{i}^{md}\rangle
+\alpha_i\langle\nabla H(B_{i}^{md}), B_{i+1}-B_{i}^{md}\rangle
+\frac{\alpha_i^2}{2}\|\mathcal{X}(B_{i+1}-B_{i})\|^2 \\
& \quad +(1-\alpha_i)P_B(B_{i}^{ag},\lambda)+\alpha_i P_B(B_{i+1},\lambda)\\
&=(1-\alpha_i)(H(B_{i}^{md})+\langle\nabla H(B_{i}^{md}), B_{i}^{ag}-B_{i}^{md}\rangle
+P_B(B_{i}^{ag},\lambda))+\alpha_i(H(B_{i}^{md})
 +\langle\nabla H(B_{i}^{md}), B_{i+1}-B_{i}^{md}\rangle)\\
& \quad+\frac{\alpha_i^2}{2}\|\mathcal{X}(B_{i+1}-B_{i})\|^2+\alpha_i P_B(B_{i+1},\lambda).
\end{split}
\end{equation}
By the convexity of $H(B)$ and \eqref{eq26a}, we have
\begin{equation}\label{eq26c}
\begin{split}
l(B_{i+1}^{ag})&=(1-\alpha_i)(H(B_{i}^{md})+\langle\nabla H(B_{i}^{md}),
 B_{i}^{ag}-B_{i}^{md}\rangle+P_B(B_{i}^{ag}))\\
&\qquad  +\alpha_i(H(B_{i}^{md})+\langle\nabla H(B_{i}^{md}), B-B_{i}^{md}\rangle)
+\alpha_i\langle\nabla H(B_{i}^{md}), B_{i+1}-B\rangle\\
&\qquad +\frac{\alpha_i^2}{2}\|\mathcal{X}(B_{i+1}-B_{i})\|^2+\alpha_i P_B(B_{i+1},\lambda)\\
& \leq (1-\alpha_i)L(B_i^{ag})+\alpha_i L(B)+\alpha_i\langle\nabla H(B_{i}^{md}), B_{i+1}-B\rangle\\
&\qquad +\frac{\alpha_i^2}{2}\|\mathcal{X}(B_{i+1}-B_{i})\|^2
 +\alpha_i P_B(B_{i+1},\lambda)-\alpha_i P_B(B,\lambda)
\end{split}
\end{equation}
Subtracting $l(B)$ from both sides of \eqref{eq26c} and rearranging some terms, we have
\begin{equation}\label{eqoptm00}
\begin{split}
&[l(B_{i+1}^{ag})-l(B)]-(1-\alpha_i)[l(B_{i}^{ag})-l(B)]\\
\leq& \alpha_i\langle\nabla H(B_{i}^{md}), B_{i+1}-B\rangle+\frac{\alpha_i^2}{2}\|\mathcal{X}(B_{i+1}-B_{i})\|^2
+\alpha_i\langle \xi, B_{i+1}-B\rangle
\end{split}
\end{equation}
where $\xi \in \partial P_B(B_{i+1},\lambda)$.
On the other hand, by the first-order optimality conditions for the sequence $B_{i+1}$ in Algorithm \ref{alg:fnsl},
we have
\begin{equation}\label{eqoptm0}
\langle \nabla H(B_{i}^{md}), B_{i+1}^e\rangle+\eta_i\langle B_{i+1}-B_i,B_{i+1}^e \rangle
+ \langle \partial P_B(B_{i+1},\lambda), B_{i+1}-B\rangle \leq 0
\end{equation}
Combining \eqref{eqoptm00} and \eqref{eqoptm0}, we obtain
\begin{equation}\label{eqoptm01}
\begin{split}
&[l(B_{i+1}^{ag})-l(B)]-(1-\alpha_i)[l(B_{i}^{ag})-l(B)]\\
\leq& \alpha_i\left\{\eta_i\langle B_{i}-B_{i+1},B_{i+1}^e \rangle+\frac{\alpha_i}{2}\|\mathcal{X}(B_{i+1}-B_{i})\|^2\right\}\\
\leq& \alpha_i\Big\{\frac{\eta_i}{2}(\|B_{i}^e\|^2-\|B_{i+1}^e\|^2-\|B_{i+1}-B_i\|^2)
 +\frac{\alpha_i}{2}\|\mathcal{X}(B_{i+1}-B_{i})\|^2\Big\}
\end{split}
\end{equation}
where we used the relationship $2\langle a-b,a-c\rangle=-\|b-c\|^2+\|a-c\|^2+\|a-b\|^2$ and the definition of $B_{i+1}^e$.

Dividing both sides of \eqref{eqoptm01} by $\alpha_i\eta_i$, we have
\begin{equation}\label{eq210}
\begin{aligned}
&\frac{1}{\alpha_i\eta_i}[l(B_{i+1}^{ag})-l(B)]
-\frac{(1-\alpha_i)}{\alpha_i\eta_i}[L(B_{i}^{ag})-L(B)]\\
\leq& \frac{1}{2}(\|B_i^e\|^2-\|B_{i+1}^e\|^2)-\frac{1}{2}(\|B_{i+1}-B_i\|^2
 -\frac{\alpha_i}{\eta_i}\|\mathcal{X}(B_{i+1}-B_{i})\|^2)\\
\leq& \frac{1}{2}(\|B_i^e\|^2-\|B_{i+1}^e\|^2)-\frac{1}{2}\Gamma_i
\end{aligned}
\end{equation}
Adding $\frac{(\beta_iQ_i+\Gamma_i)}{2}$ to both sides of \eqref{eq210}, we have
\begin{equation}\label{eq2100}
\begin{aligned}
&\frac{1}{\alpha_i\eta_i}[l(B_{i+1}^{ag})-l(B)]
-\frac{(1-\alpha_i)}{\alpha_i\eta_i}[l(B_{i}^{ag})-l(B)]
+\frac{(\beta_iQ_i+\Gamma_i)}{2}\\
\leq& \frac{1}{2}(\|B_i^e\|^2-\|B_{i+1}^e\|^2)+\frac{(\beta_i-1)Q_i}{2}+\frac{Q_i}{2}
\end{aligned}
\end{equation}
Since $Q_{i+1}=\beta_iQ_i+\Gamma_i$, $0\leq \beta_i \leq (1-\frac{1}{i})^2$, and $Q_{i}\geq -\frac{C}{(i-1)^2}$, we obtain
\begin{equation}\label{eq2101}
\begin{aligned}
&\frac{1}{\alpha_i\eta_i}[l(B_{i+1}^{ag})-l(B)]
-\frac{(1-\alpha_i)}{\alpha_i\eta_i}l(B_{i}^{ag})-l(B)]
+\frac{Q_{i+1}}{2}\\
\leq& \frac{1}{2}(\|B_i^e\|^2-\|B_{i+1}^e\|^2)+\frac{(1-\beta_i)C}{2(i-1)^2}+\frac{Q_i}{2}
\end{aligned}
\end{equation}
Setting $B=\hat{B}$, by
the relationship $\frac{1}{\alpha_i\eta_i}=\frac{1-\alpha_{i+1}}{\alpha_{i+1}\eta_{i+1}}$,
and $\alpha_1=1$, we obtain
\begin{equation}\label{eq2102}
\frac{1}{\alpha_i\eta_i}[l(B_{i+1}^{ag})-l(B)]
\leq \frac{1}{2}\|B_0-\hat{B}\|^2+\sum_{i=2}^k\frac{(1-\beta_i)C}{(i-1)^2}+\frac{C}{k^2}
\end{equation}
after summing \eqref{eq2101} from $i = 1$ to $k$.

Next we show the upper bound of $\alpha_k\eta_k$.
Since $\eta_{\min}\leq \eta_{0,1}$, we have $\eta_{\min}\leq \|\mathcal{X}^T\mathcal{X}\|_2$. Then, by definition of
$\eta_{0,i}$, we get
\begin{equation}\label{eq1a00}
\eta_{\min}\leq \eta_{0,i}\leq ||\mathcal{X}^T\mathcal{X}||_2.
\end{equation}
Denote $\sigma^l\eta_{0,i}$ by $\eta_i^{'}$, where $l$ is the number of line search in Step 3 of Algorithm \ref{alg:fnsl}. By $\frac{1}{\alpha_i\eta_i}=\frac{1-\alpha_{i+1}}{\alpha_{i+1}\eta_{i+1}}$ and the definition of $\eta_{i}$,
we have
\begin{equation}\label{eq1a0}
\frac{1}{\alpha_i\sqrt{\eta_i^{'}}}=\frac{\sqrt{1-\alpha_{i+1}}}{\alpha_{i+1}\sqrt{\eta_{i+1}^{'}}}
\leq  \frac{1}{\alpha_{i+1}\sqrt{\eta_{i+1}^{'}}}-\frac{1}{2\eta_{i+1}^{'}}  \ \  \text{for}\ \ i\geq 1
\end{equation}
Then, by induction we can get, with $\alpha_1=1$,
$$(\frac{1}{\sqrt{\eta_1^{'}}}+\frac{1}{2}\sum_{i=2}^{k}\frac{1}{\sqrt{\eta_k^{'}}})^2\leq \frac{1}{\alpha_k^2\eta_k^{'}}$$
which implies
\begin{equation}\label{eq1a}
\alpha_k\eta_k\leq \frac{1}{(\frac{1}{\sqrt{\eta_1^{'}}}+\frac{1}{2}\sum_{i=2}^{k}\frac{1}{\sqrt{\eta_k^{'}}})^2}
\leq  \frac{4\sigma||\mathcal{X}^T\mathcal{X}||_2}{(k+1)^2}   \ \ \text{for} \ \ k\geq 1
\end{equation}
where we used \eqref{eq1a0} and the definition of $\eta_i^{'}$.

Combining \eqref{eq2102} and \eqref{eq1a}, we obtain \eqref{TR1}.
\end{proof}

\section{Minimizer of optimization problems}

{For convenience, we provide solutions for the minimization problems in Algorithm 1 when the penalty term has different norms.

When $P(B,\lambda) = \lambda\|B\|_1$,
$$B_{i+1} = S(B_i - \frac{\lambda}{\eta_i}\mathcal{X}^T(\mathcal{X}B_i^{md}-\mathcal{Y}))$$
where $S(B)_{kl} = \max(0, |b_{kl}|-\frac{\lambda}{\eta_i})\text{sgn}(b_{kl})$ and $b_{kl}$ is the entry of matrix $B$ from $k$th row and $l$th column.

When $P(B,\lambda) = \lambda\|B\|_{2,1}$,
$$B_{i+1} = GS(B_i - \frac{\lambda}{\eta_i}\mathcal{X}^T(\mathcal{X}B_i^{md}-\mathcal{Y}))$$
where $GS((B)_{G_k}) = \max(0, \frac{\lambda/\eta_i}{\|(B)_{G_k}\|_2})((B)_{G_k})$

When $P(B,\lambda) = \lambda\|B\|_{\ast}$,
$$B_{i+1} = D_{\tau}(B_i - \frac{\lambda}{\eta_i}\mathcal{X}^T(\mathcal{X}B_i^{md}-\mathcal{Y}))$$
where $D_{\tau}(Z) = UD_{\tau}V^{\ast}$, which is singular value decomposition. $D_{\tau}=\text{diag}(\max(0,\sigma_i-\lambda/\eta_i))$
and $\{\sigma_i\}_{i=1}^{k}$ are the singular values of matrix $Z$}.

\section{Algorithm 2}
Due to space limitations, the detailed Algorithm 2 is
given next.
\begin{algorithm*}
	\caption{Adaptive Fast Network Structure Learning (AFNSL) method}
    \label{alg:AFNSL}
	\begin{algorithmic}
		\STATE Choose $C\geq 0, \sigma > 1, \eta_{0,1} \geq\eta_{\min}$.
               Set $\alpha_1=1, L_{1}^{ag}=L_1,R_{1}^{ag}=R_1$, and $Q_1=0$.
		\STATE \textbf{For} $i = 1,2, \ldots, k$,
		\begin{enumerate}[\hspace{.2cm}1.\hspace{.5cm}]
			\item[] \slash\slash {\it { Backtracking }}
			\item Set $\eta_i=\alpha_i\eta_{0,i}$, where $\eta_{0,i}$ is from 14. Solve $\alpha_{i}$ from
                 $\frac{1}{\alpha_{i-1}\eta_{i-1}}=\frac{1-\alpha_{i}}{\alpha_{i}\eta_{i}}$ for $i>1$. Compute
			\begin{align*}
            L_i^{md}     = & (1 - \alpha_i)L_i^{ag} + \alpha_i L_i,
            \\
			R_i^{md}     = & (1 - \alpha_i)R_i^{ag} + \alpha_i R_i,
			\\
			L_{i+1} = & \underset{L \in \Omega}{\arg\min}\left\{\langle \nabla l(L_i^{md},R_i^{md}), L\rangle+\frac{\eta_i}{2}\|L-L_i\|^{2}_F
			+ \lambda_N\|L\|_{*}\right\},
            \\
			R_{i+1} = & \underset{R}{\arg\min}\left\{\langle \nabla l(L_i^{md},R_i^{md}), R\rangle+\frac{\eta_i}{2}\|R-R_i\|^{2}_F
			+ \mu_N\|R\|_{\diamond} \right\},
            \\
			\Gamma_i=&\|L_{i+1}+R_{i+1}-L_{i}-R_{i}\|^2-\frac{\alpha_i}{\eta_i}(\|\mathcal{X}(L_{i+1}+R_{i+1}-L_{i}-R_{i}\|^2_F),
			\\
			Q_{i+1}=& \beta_iQ_i+\Gamma_i,\ \text{ where }0\leq \beta_i \leq (1-\frac{1}{i})^2.
			\end{align*}
			\item
			If $Q_{i+1}< -{C}/{i^2}$, then replace $\eta_{0,i}$ by $\sigma\eta_{0,i}$ and return to step 1.
			
			\item[] \slash\slash {\it {Updating iterates}}
			
			\item Compute
			\begin{align*}
            L_{i+1}^{ag} = &(1-\alpha_i)L_i^{ag}+\alpha_iL_{i+1},
            \\
            R_{i+1}^{ag} = &(1-\alpha_i)R_i^{ag}+\alpha_iR_{i+1}.
            \end{align*}
		\end{enumerate}
		\STATE \textbf{EndFor}
		\STATE \textbf{Output} $(L_{k+1}^{ag}, R_{k+1}^{ag})$.
	\end{algorithmic}
\end{algorithm*}

\section{Additional Numerical Experiments}

\subsection*{Sparse plus group-sparse network learning problem}\label{subsec:splusgs}
\begin{table}[t!]
\centering
\begin{tabular}{|c|c|c|c|c|c|}
\hline
p                 & N                 & model  & (TPR, FAR)(\%) & EE \\ \hline
\multirow{6}{*}{50} & \multirow{3}{*}{100} & Lasso   &  (67.9, 16.0)  & 0.40  \\ \cline{3-5}
                 &                   & SGL & (71.1, 18.4)   & 0.41 \\ \cline{3-5}
                  &                   & S+G   &   ($\mathbf{86.7}$, $\mathbf{16.5}$) &  $\mathbf{0.39}$    \\ \cline{2-5}
                  & \multirow{3}{*}{200} & Lasso   &  (69.3, 10.4)    & 0.32  \\ \cline{3-5}
                  &                   & SGL & (71.1, 11.6) & 0.33 \\  \cline{3-5}
                  &                   & S+G   & ($\mathbf{90.4}$, $\mathbf{10.4}$)  & $\mathbf{0.28}$   \\ \hline
\multirow{6}{*}{100} & \multirow{3}{*}{100} & Lasso   &  (70.2, 21.2) & 0.44  \\ \cline{3-5}
                 &                   & SGL  & (71.6, 23.5)&  0.46 \\ \cline{3-5}
                  &                   & S+G  &  ($\mathbf{84.7}$, $\mathbf{21.2}$) & $\mathbf{0.43}$ \\ \cline{2-5}
                  & \multirow{3}{*}{200} & Lasso   & (87.8, 17.4) & 0.34  \\ \cline{3-5}
                  &                   & SGL & (77.7, 19.7)   &  0.35\\ \cline{3-5}
                  &                   & S+G  & ($\mathbf{87.8}$, $\mathbf{17.4}$)   & $\mathbf{0.32}$ \\ \hline
\multirow{6}{*}{200} & \multirow{3}{*}{100} & Lasso   &  (75.2, 43.4)  & 0.76  \\ \cline{3-5}
                 &                   & SGL & (76.7,  43.3)  &  0.76\\ \cline{3-5}
                  &                   & S+G   & ($\mathbf{84.3}$, $\mathbf{43.7}$)   & $\mathbf{0.74}$   \\ \cline{2-5}
                  & \multirow{3}{*}{200} & Lasso   & (70.1, 23.1)  & 0.55  \\ \cline{3-5}
                 &                   & SGL &   (71.4, 23.5)  &  0.55 \\ \cline{3-5}
                  &                   & S+G   & ($\mathbf{78.5}$, $\mathbf{23.5}$)  & $\mathbf{0.54}$ \\ \hline
\end{tabular}
\caption{Performance comparison of S+G with Lasso and SGL on sparse plus group-sparse network identification problem.}
\label{my-tabelsgs}
\end{table}

We discuss the experimental setting for a
sparse plus group-sparse transition matrix and compare the performance with methods that either assume pure sparsity (lasso) or pure group sparsity (group lasso, SGL).

We consider three different VAR(1) models with $p=50, \, 100$ and $200$ variables. For each of these models, we generate $N=100$ and $200$ observations from a Gaussian VAR(1) process (see equation (9) in the main file) where $B$ can be decomposed into a sparse matrix $S$ with $5\%$ non-zero entries and a group-sparse matrix $G$ with each column corresponding to a different group (hence we have $p$ groups in $G$).
We randomly select two columns (groups) consisting of two super hubs, in which the strength of the edges is generated independently from a Gaussian distribution. To be more consistent with our error bound analysis, we set $\gamma$ to $p/2$ in \eqref{eqn:opt-sparse-lowrank}. The network topology of $S$ is generated the same way as that in subsection \ref{subsec:sparse} except that the occurring probability of the edge from one node to another node $\xi$ is set to be $0.05$. Subsequently, the entries of the corresponding two columns in $G$ are set to be zero. Finally, we rescale the entries of $B$ so that a desired spectral radius is reached. We employ the TPR, FAR, and EE metrics in the comparisons.

%
Table \ref{my-tabelsgs} shows the experimental results for different network size and number of samples. It can be seen that utilizing an S+G model enables us to identify a larger portion of correct nonzero numbers in $B$, while achieving almost the same false alarm rate compared to lasso and SGL. Particularly, the S+G model can recover the group information perfectly while lasso and SGL miss a number of edges, as expected since they correspond to misspecified structures. Further, the
S+G model exhibits the lowest estimation error amongst them. It should be noted that the advantage of the S+G model will be more evident if the strength of the edges within the groups is weaker.

In Supplement Figure 1 the true network structure of S+G, S and G with $p=50$ and $N=200$ is depicted. The recovered network structures by the S+G model are given in Supplement Figure 2 (top) and the group-sparse components are recovered perfectly. We also compare the recovered network structures by S+G, SGL and lasso models, see Supplement Figure 2 (bottom), from which we can see that S+G performs best.


\begin{figure}[!tbp]
      \centering
      {\includegraphics[width=\textwidth]{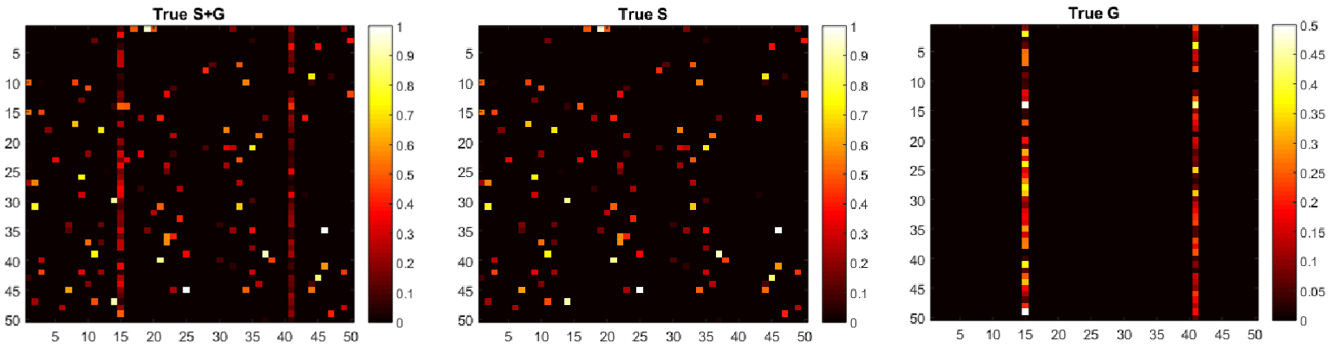}} 
\caption{True network structure of S+G, S and G with $p=50$ and $N=200$.}
\label{fig:true-transition-matrix-splusgs}
\end{figure}

\begin{figure}[!tbp]
      \centering
      {\includegraphics[width=\textwidth]{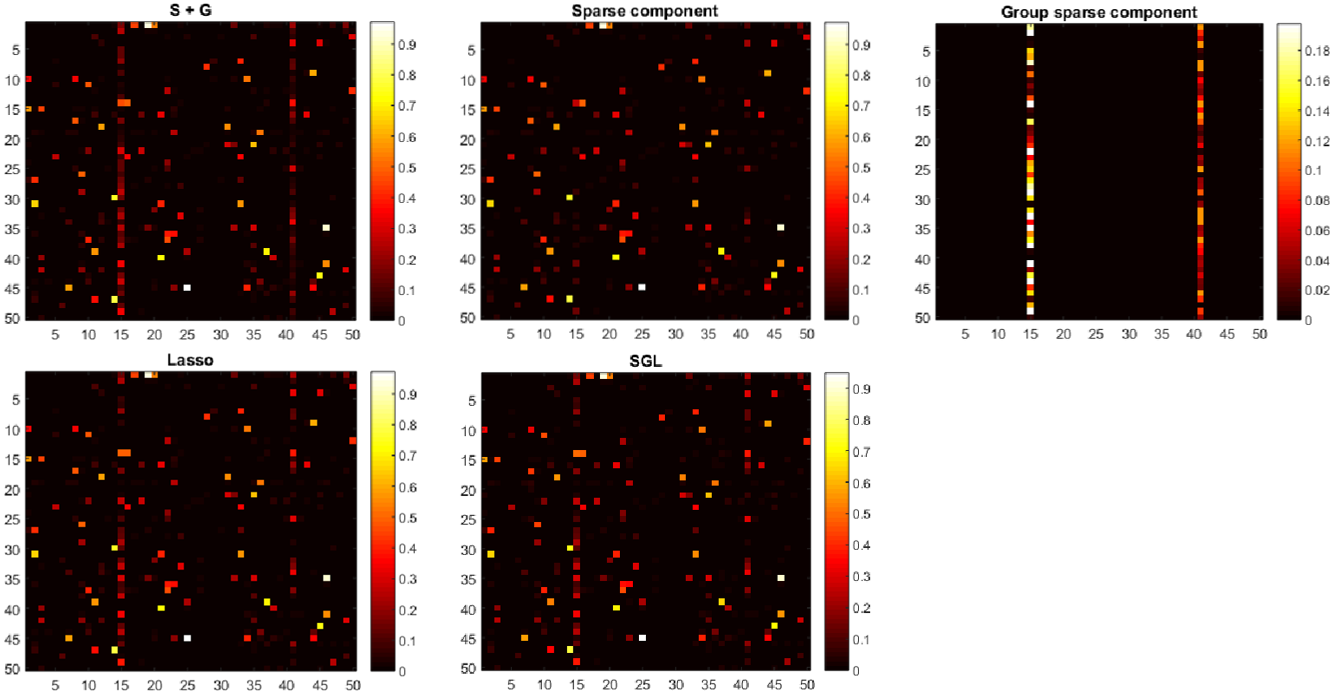}} 
\caption{Network structure identified by $\hat{S}+\hat{G}$ (top) (with its sparse and group-sparse components), Lasso and SGL (bottom), respectively.}
\label{fig:Est-transition-matrix-splusgs}
\end{figure}

Figure \ref{fig:true-transition-matrix-splusgs} shows the true network structure
S+G, S and G with $p=50$ and $N=200$. The recovered network structures by S+G model are given in Figure \ref{fig:Est-transition-matrix-splusgs} (top). Clearly, the group-sparse components are recovered perfectly. We also compare the recovered network structures by S+G, SGL and lasso models, see Figure \ref{fig:Est-transition-matrix-splusgs} (bottom), from which we can see that S+G performs best.

Figures \ref{fig:transition-matrix-lr-sparse-groupsparse} shows the estimated Granger causal network using low-rank plus sparse plus group-sparse VAR estimates using a VAR(1) model with $p = 50$ and $n = 300$. The top panel of the Figures \ref{fig:transition-matrix-lr-sparse-groupsparse} displays the true the structure of sparse plus group-sparse components $S+G$, the structure of group-sparse component $G$, and the structure of low-rank component $L$. The bottom panel of the Figures \ref{fig:transition-matrix-lr-sparse-groupsparse} displays the
structure of the Granger causal networks estimated by the method of $L+S$ and $S+G$ modeling strategy. It can be seen that the $S+G$ estimate selects many false positives due to its failure to account for the latent structure. On the other hand, the $L+S$ method provides
an estimate exhibiting significantly fewer false positives entries as that by $L+S+G$.

\begin{figure}[!tbp]
      \centering
      \includegraphics[width=\textwidth]{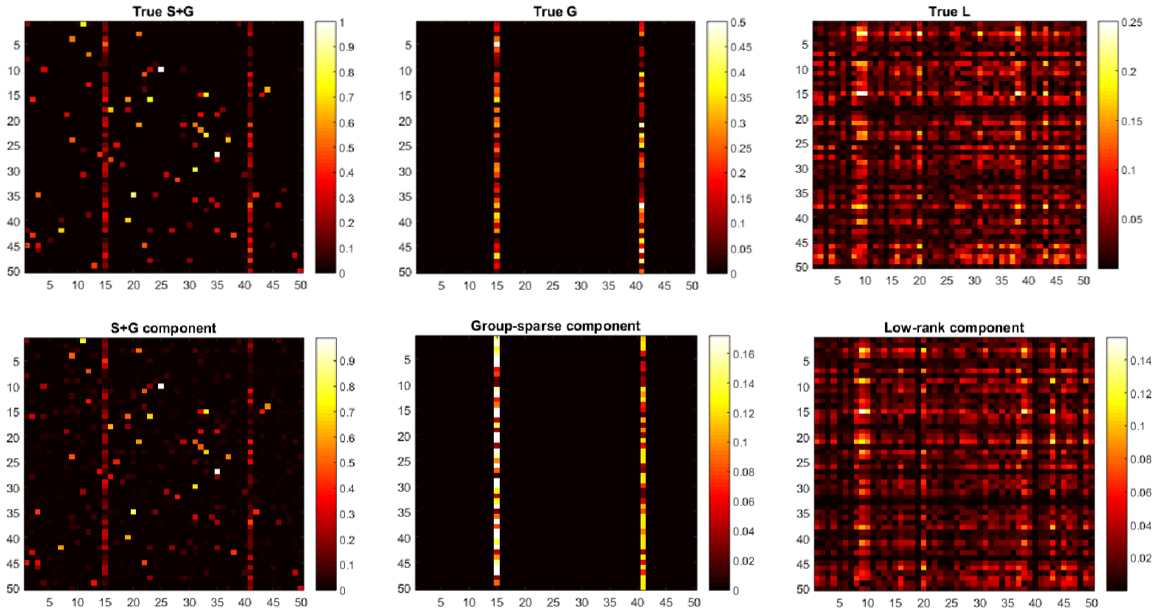} 
\caption{Estimated Granger causal networks using low-rank plus sparse plus group-sparse VAR estimates. The top panel displays the true structure of the sparse plus group-sparse components $S+G$, the group-sparse component $G$, and low-rank component $L$ of the true transition matrix $B$.  The bottom panel displays the structure of the Granger causal networks estimated by L+S+G, the estimated sparse plus group-sparse components ($\hat{S}+\hat{G}$), the estimated group-sparse component ($\hat{G}$), and the estimated low-rank component ($\hat{L}$).}
\label{fig:transition-matrix-lr-sparse-groupsparse}
\end{figure}

\begin{figure}[!tbp]
      \centering
      \includegraphics[width=0.7\textwidth]{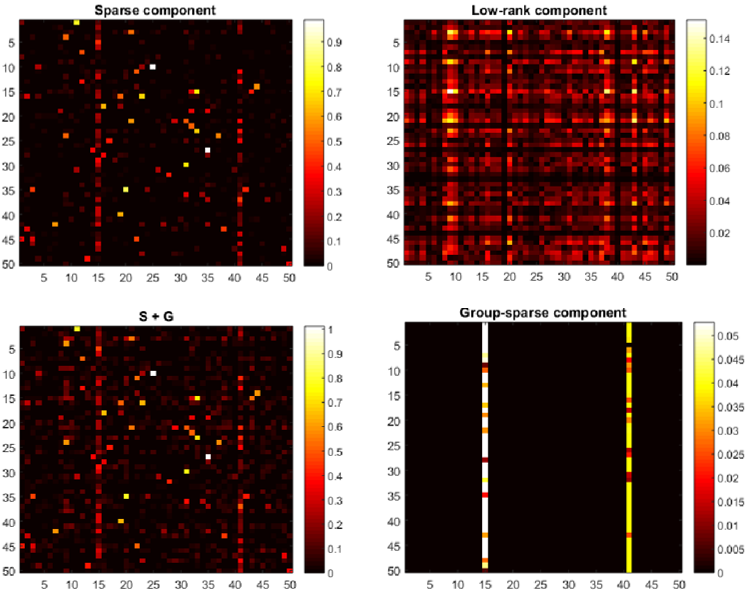} 
\caption{Estimated Granger causal networks using low-rank plus sparse and  group-sparse VAR estimates. The top panel displays the structure of the Granger causal networks estimated by L+S, the estimated sparse component ($\hat{S}$) and the estimated low-rank component ($\hat{L}$). The bottom panel displays the structure of the Granger causal networks estimated by S+G, the estimated sparse plus group-sparse components ($\hat{S}+\hat{G}$) and the estimated group-sparse component ($\hat{G}$).}
\label{fig:transition-matrix-lr-sparse-groupsparse}
\end{figure}

The learned Granger causal network structures estimated by lasso are given in figure \ref{fig:transition-matrix-lasso-graph}, which correspond to the case of pre-crisis, crisis and post-crisis, respectively. From the network structures by Lasso and S+L strategy, we can see that companies are highly connected when financial crisis is coming. After accounting for the low-rank component, the estimated sparse components by L+S strategy are usually more sparse than that by Lasso.

\begin{figure}[!tbp]
\centering
\begin{subfigure}{.33\textwidth}
  \centering
  \includegraphics[page=1,width=\linewidth]{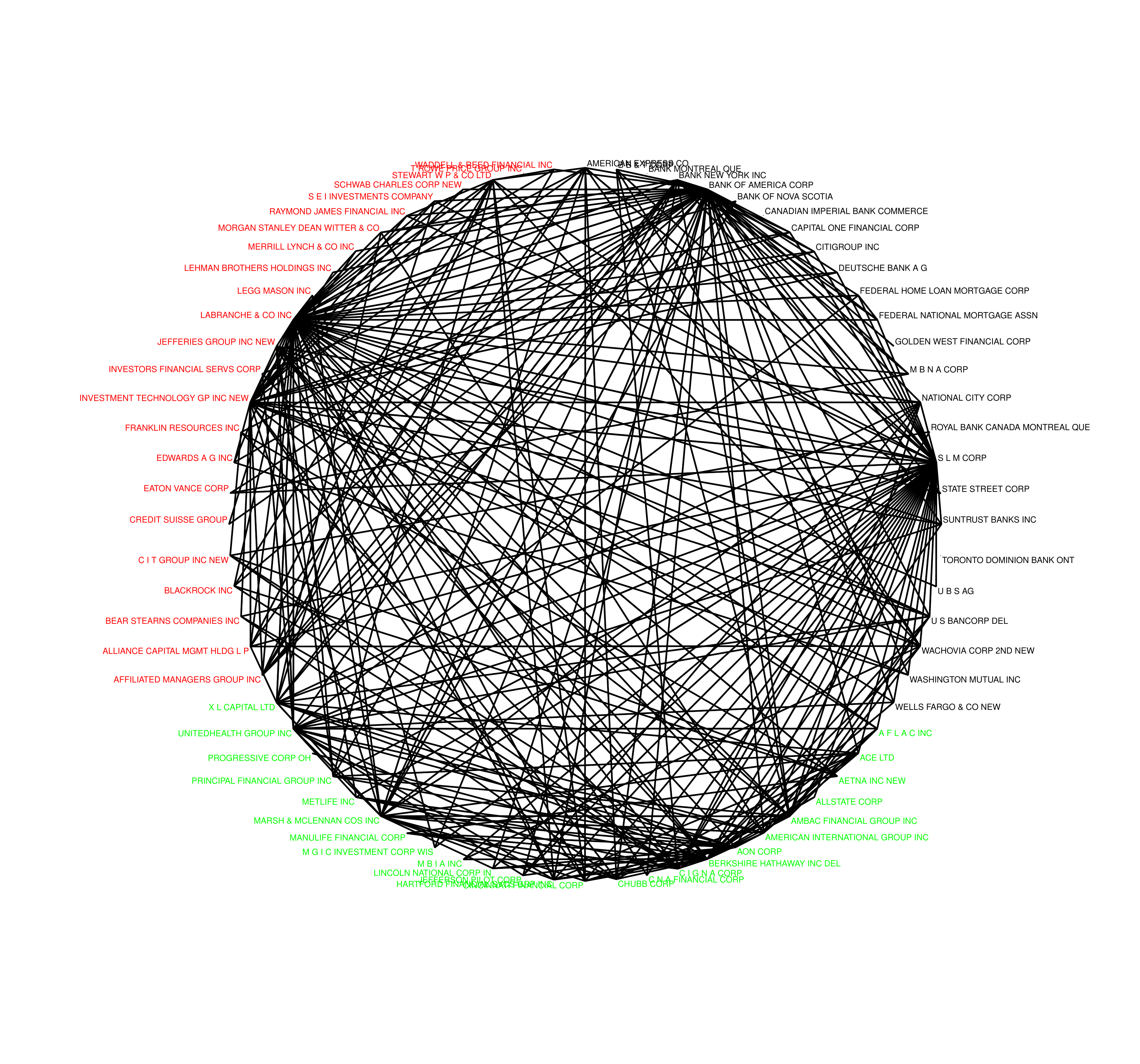}
\end{subfigure}%
\begin{subfigure}{.33\textwidth}
  \centering
  \includegraphics[page=1,width=\linewidth]{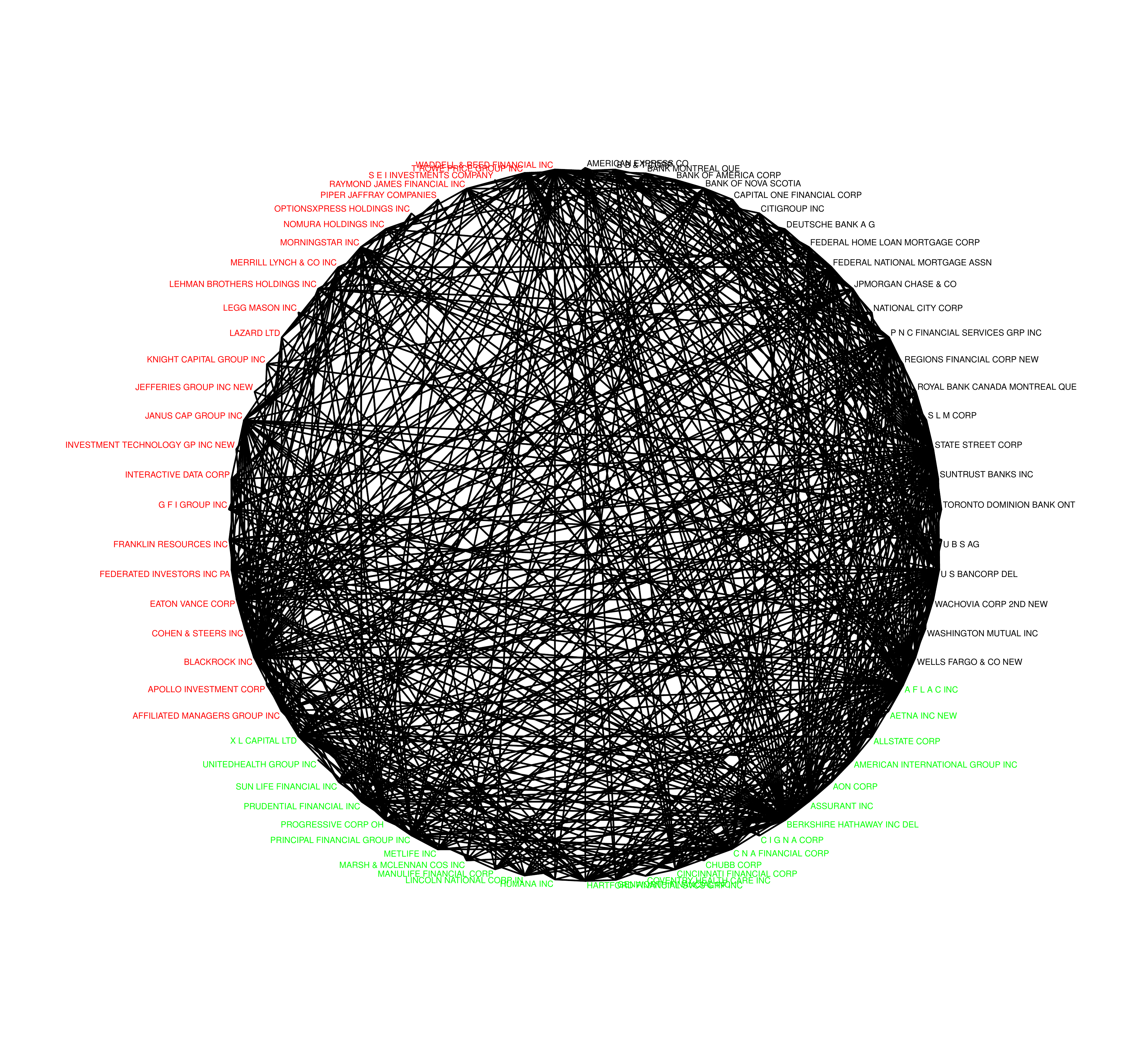}
\end{subfigure}
\begin{subfigure}{.33\textwidth}
  \centering
  \includegraphics[page=1,width=\linewidth]{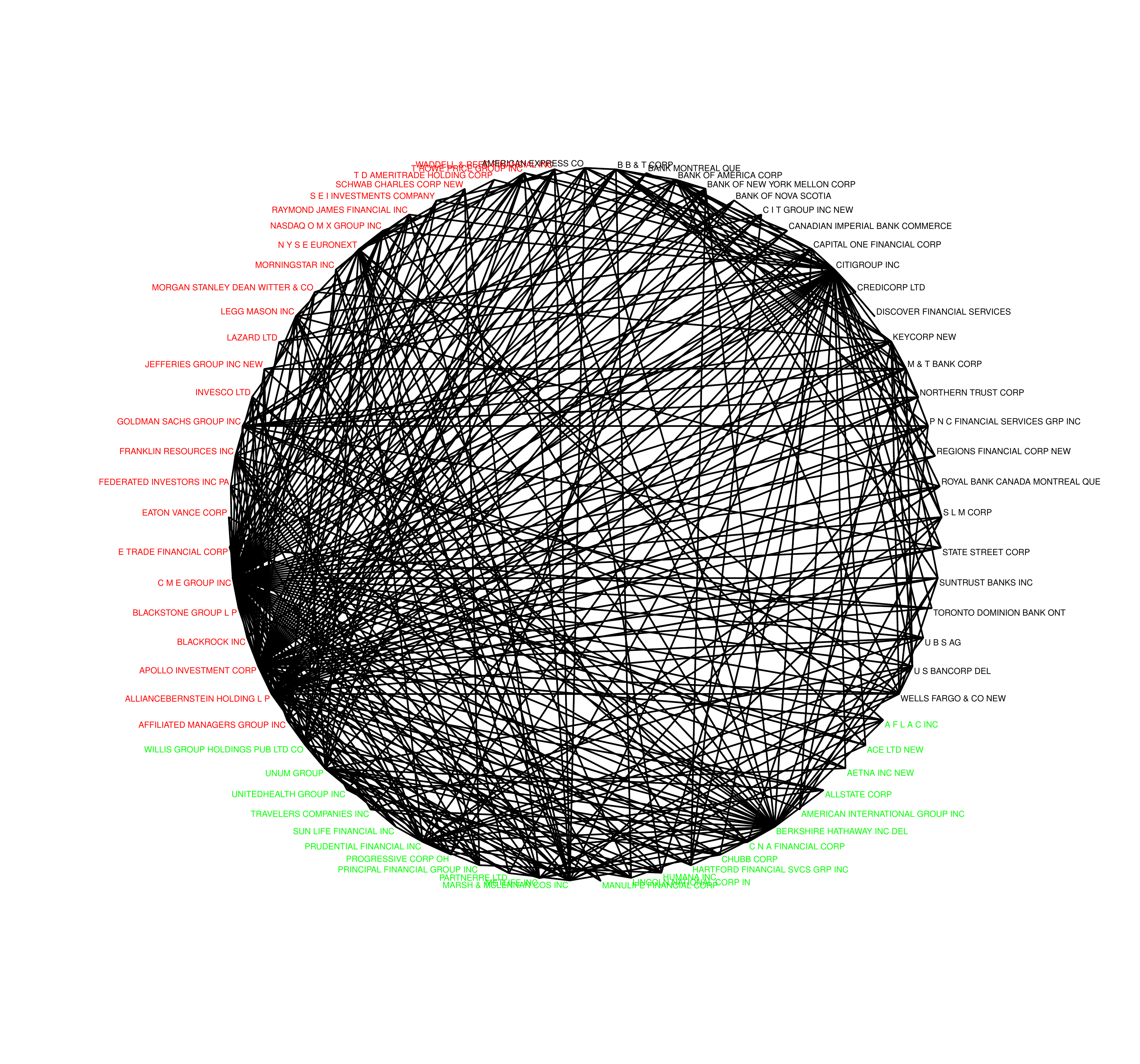}
\end{subfigure}
\caption{Graph layout of Pre-crisis case, Crisis case and Post-crisis estimated by lasso VAR estimates. The left displays the structure of the Granger causal networks of the estimated $\hat{B}$ with 383 nonzero entries. The middle displays that of the estimated $\hat{B}$ with 801 nonzero entries. The right displays that of the estimated $\hat{B}$ with 477 nonzero entries.}
\label{fig:transition-matrix-lasso-graph}
\end{figure}

\newpage


\end{document}